\documentclass[11]{llncs}
\usepackage{fullpage}

\usepackage[parfill]{parskip}
\usepackage[cmex10]{amsmath}
\usepackage{xspace}
\usepackage{xcolor,graphicx}

\usepackage{amsthm}
\usepackage{amssymb}
\usepackage{epstopdf}  
\usepackage{rotating}
\usepackage{setspace}

\usepackage{dsfont}
\definecolor{TSUYUKUSA}{RGB}{46, 169, 223}
\definecolor{KURENAI}{RGB}{203, 27, 69}
\usepackage[pdfstartview=FitH,pdfpagemode=UseNone,colorlinks=true,citecolor=KURENAI,linkcolor=TSUYUKUSA,linktocpage=true]{hyperref}
\usepackage{subfigure}
\usepackage{algpseudocode}
\usepackage{algorithmicx}
\usepackage{fancyhdr}
\usepackage{mdframed}
\usepackage{float}
\usepackage{tikz}
\usetikzlibrary{shapes.geometric}
\usepackage{rotating}	
\usepackage{enumitem}
\usepackage[capitalise]{cleveref}
\usepackage[nottoc]{tocbibind}
\usepackage{appendix}
\usepackage{authblk}
\usepackage{doi,url}
\usepackage{breakcites}
\usepackage{qcircuit}
\usepackage{cryptocode}
\usepackage{caption} 
\usepackage[sectionbib, numbers]{natbib}

\makeatletter
\renewcommand{\subsection}{\@startsection{subsection}{2}{\z@}%
                      {8\p@ \@plus -4\p@ \@minus -4\p@}%
                      {8\p@ \@plus 2\p@ \@minus 2\p@}%
                      {\normalfont\large\bfseries\boldmath
                        \rightskip=\z@ \@plus 8em\pretolerance=10000 }}
\renewcommand\subsubsection{\@startsection{subsubsection}{3}{\z@}%
                       {-18\p@ \@plus -4\p@ \@minus -4\p@}%
                       {4\p@ \@plus 2\p@ \@minus 2\p@}%
                       {\normalfont\normalsize\bfseries\boldmath
                        \rightskip=\z@ \@plus 8em\pretolerance=10000 }}
\makeatother

\input{macros}

\begin{document}

\title{Generalized Hybrid Search and \\ Applications to Blockchain and Hash Function Security}

\author{Alexandru Cojocaru\inst{1} \and Juan Garay\inst{2} \and Fang
  Song\inst{3}}

\institute{QuICS, University of Maryland, USA.\\
  \email{cojocaru@umd.edu}
  \and Department of Computer Science and
  Engineering, Texas A\&M University, USA.\\
  \email{garay@tamu.edu}
  \and
  Department of Computer Science, Portland State
  University, USA.\\
  \email{fang.song@pdx.edu}
}

\maketitle
\setcounter{page}{0}
\thispagestyle{empty}

\begin{center}
\textbf{Abstract}    
\end{center}

\begin{changemargin}{1.5cm}{1.5cm}
{ \small
In this work
we first examine the hardness of solving various search problems by hybrid quantum-classical strategies, namely, by algorithms that have both quantum and classical capabilities. 
We then construct a hybrid quantum-classical search algorithm and analyze its success probability.

Regarding the former,  
for search problems that are allowed to have multiple solutions and in which the input is sampled according to
arbitrary distributions
we establish their hybrid quantum-classical query complexities---i.e., given a fixed number of classical and quantum queries, determine what is the probability of solving the search task. At a technical level, our results generalize the 
framework for hybrid quantum-classical search algorithms recently proposed by Rosmanis~\cite{Ros22}.
Namely, for an {\em arbitrary} distribution $D$ on Boolean functions, the probability that an algorithm equipped with $\tau_c$ classical queries and $\tau_q$ quantum queries succeeds in finding a preimage of $1$ for a function sampled from $D$ is at most $\solnl\cdot
  (2\sqrt{\cquery} + 2\qquery + 1)^2$, where $\solnl$ captures the average (over $D$) fraction of preimages of $1$.

As applications of our hardness results,
we first revisit and generalize the formal security treatment of the Bitcoin 
protocol called the {\em Bitcoin backbone} [Eurocrypt 2015], to a setting where the adversary has both quantum and classical capabilities, 
presenting a new {\em hybrid honest majority} condition necessary for the protocol to properly operate.
Secondly, we re-examine the generic security of hash functions
[PKC 2016] against quantum-classical hybrid adversaries. 

Regarding our second contribution, we design a hybrid algorithm which first spends all of its classical queries 
and in the second stage runs a ``modified Grover'' in which the initial state depends on the target distribution $\distr$. We then show how to analyze its success probability for arbitrary target distributions and, importantly,
its optimality for the uniform and the Bernoulli distribution cases. 
}
\end{changemargin}

\makeatletter
\renewcommand*\l@author[2]{}
\renewcommand*\l@title[2]{}
\makeatletter
\setcounter{tocdepth}{2}
\setcounter{page}{1}

\newpage

\tableofcontents

\newpage

\section{Introduction}
\label{sec:intro} 
The query model is an elegant abstraction and is widely adopted in
cryptography. A notable example is the random oracle (RO)
model~\cite{BR93}, where a hash function $f$ is modeled as a random
black-box function, and all 
parties including the adversary can
evaluate it only by issuing a query $x$ and receiving $f(x)$ in
response. Numerous cryptosystems have been designed and analyzed in
the random oracle model (e.g.,~\cite{BR94,BR96,Shoup01,FOPS04,FO-JoC13}). 

Quantum computing brings about a new quantum query model, where
\emph{superposition} queries to the hash function $f$ are permitted in
the form of: 
$\sum_{x,y} \alpha_{x,y} \ket{x} \ket{y} \mapsto \sum_{x,y}
\alpha_{x,y} \ket{x}\ket{y\oplus f(x)}$. This equips quantum
adversaries with new capabilities. Indeed, some classically secure
digital signature and public-key encryption schemes are broken in the
\emph{quantum} random oracle (QRO) model, where a quantum adversary 
makes superposition queries to $f$~\cite{YZ21}. A 
significant amount of
effort has 
been devoted to address such quantum-query
adversaries (cf.,~\cite{BDF+11,ES15,Unruh15,HHK17,AHU19,DFMS19,CMS19,ES20,DFMS22})
and often, in order to maintain security, we need to pay a considerable
efficiency overhead, such as more complex constructions or larger key
sizes.

The threat is alarming, but it requires running a large-scale
quantum computer coherently for an extended time. The quantum devices
available in the near-to-intermediate term are likely to be computationally restricted as
well as expensive~\cite{Preskill18}. This reality inspires a
\emph{hybrid} query model, where one is granted a quota of both
classical and quantum queries,
a model
which subsumes
the classical and fully
quantum query models
as special cases. Establishing a
trade-off between classical and
quantum queries allows us to give a more accurate estimation of
security and hence optimized parameter choices of a cryptosystem
depending on what resources are available to a (near-term) quantum
adversary.

Recently, Rosmanis studied
the basic unstructured
search problem in the hybrid query model~\cite{Ros22}, where given oracle
function $f: X\to \bits$, one wants to find a ``marked'' input, i.e.,
$x$ with $f(x) = 1$. This search problem and many variants, such as  
 multiple or randomly chosen marked inputs, are well understood when
all queries are quantum~\cite{Grover96,BBBV97,Zalka99,DH09,Zhandry19}, where Grover's quantum algorithm
gives a quadratic speedup over classical algorithms and it is also
proven optimal. Rosmanis's work proves
the hardness of searching in the domain of a function with a \emph{unique}
marked input $x^*$ in the hybrid query model. Specifically,
any quantum algorithm with $\cquery$ classical queries and $\qquery$
quantum queries succeeds in finding $x^*$ with probability at most
$ \frac{1}{|X|} \cdot (2 \sqrt{\cquery} + 2\qquery + 1)^2$.
This hardness bound is also shown in~\cite{HLS22} by a new recording technique tailored to the hybrid query model.

\medskip

\subsection{Our Contributions and Technical Overview}
\label{sec:contr}


\paragraph{Success Probability Lower Bound 
of any Hybrid Search Algorithm.}
In this work, we consider an arbitrary distribution $\distr$ on the function
family $\func = \{f : X\to \bits\}$,
and prove a precise upper
bound on the probability of finding a preimage $x$ with $f(x)=1$ when
$f\gets \distr$, for any algorithm $\adv$ spending $\cquery$
classical and $\qquery$ quantum queries. 
Specifically, we show that:
\begin{equation*}
  \Pr_{f\gets \distr}[ f(x)=1: x\gets \adv^{{f}} ] \le \solnl\cdot
  (2\sqrt{\cquery} + 2\qquery + 1)^2 \, , 
\end{equation*}  
where $\solnl \defeq \sup_{\varphi : \norm{\varphi} \leq 1} \left( \expt_{f\gets \distr} 
\norm{ \left(\sum_{x: f(x)=1} 
\ket{x}\bra{x}\right) \varphi}^2  \right)$ captures the \emph{average} fraction of preimages of $1$, and is solely determined by the distribution $\distr$.

With our generalized bound, deriving hardness bounds for specific distributions becomes 
convenient. All we need is to analyze
$\solnl$, and this usually can be done by simple combinatorial
arguments. For instance, let $\distr$ be the uniform distribution over
functions with exactly one marked input. Then we can observe that $\solnl = \Pr_{f \gets \distr }[f(x) = 1] = {1}/{|X|}$ for an arbitrary $x$, which
recovers the result of Rosmanis~\cite{Ros22}. 
The hardness of searching 
a function with $w > 1$ marked
items 
can be 
similarly 
derived.

We further demonstrate our result on another distribution
$\distr_\berpar$, where each input is marked according to a Bernoulli
trial. Namely, for every $x \in X$, we set $f(x) = 1$ with probability
$\berpar$ \emph{independently}. By determining $\solnl$ in this case,
we derive the hardness of search when the function is drawn from $\distr_\berpar$.
This search problem under $\distr_\berpar$, which we call {\em Bernoulli
Search}, is particularly useful in 
several cryptographic
applications. Firstly, we can prove generic security bounds for 
hash function properties,
such as preimage-resistance,
second-preimage resistance and their multi-target extensions, against
hybrid quantum-classical adversaries. 
This follows by adapting the reductions in~\cite{HRS16}, where the hash properties are connected to the 
{\em Bernoulli Search} problem 
in the fully quantum query setting and then
plugging in
our hybrid hardness bound of {\em Bernoulli Search}. In another application,
{\em Bernoulli Search} was shown to dictate the security of 
proofs of work (PoWs)
and security properties of Bitcoin-like
blockchains 
in the random oracle model
(with fully quantum queries)~\cite{CGKSW23}. This allows us to
identify a new \emph{honest-majority} condition under which the
security of Bitcoin blockchain holds against hybrid adversaries with 
classical and quantum queries.

At a technical level, the proof of our hardness bound follows the overall strategy
of~\cite{Ros22}. As in the standard optimality proof of Grover's
algorithm~\cite{BBBV97}, one would consider running an adversary's
algorithm with respect to the input function $f\gets \distr$ or a
constant-0 function. Then one argues that each query diverges the
states in these two cases, which is called a \emph{progress measure},
by a small amount. On the other hand, in order to find a marked input
in $f$, the final states need to differ significantly. Therefore,
sufficiently many queries are necessary for the cumulative progress to
grow adequately.

Now, when classical queries are mixed with quantum queries, the quantum
states would collapse after each classical query and it becomes
unclear how to measure the progress. To address this, Rosmanis 
considers instead an intermediate oracle named
\emph{pseudo-classical}. Namely, consider a quantum query with the
output register initialized in $\ket{0}$:
$\sum_{x} \alpha_{x} \ket{x} \ket{0} \mapsto \sum_{x} \alpha_{x}
\ket{x}\ket{f(x)}$. We can then view a classical query as the result
of measuring the input register that collapses to $x$ and receiving
$f(x)$, whereas a pseudo-classical oracle measures the output
register, resulting in one of two possible outcomes:
$\sum_{x: f(x) = 0} \alpha_x \ket{x}\ket{0}$ (denoted as the 0-outcome branch) or
$\sum_{x: f(x) = 1} \alpha_x \ket{x}\ket{1}$ (denoted as the 1-outcome branch). With this change, one
instead tracks the progress between the 0-outcome branch in case of
$f\gets \distr$ and 
the state in case of the constant-0 function (which always stays in the 0-outcome branch). The algorithm fails if its state stays
in the 0-outcome branch and is close to the state in the constant-0
case. A key ingredient in our proof is to deliberately separate the evolution of
various objects on an \emph{individual} function and what
\emph{characteristics} of the distribution $\distr$ influence the
evolution and in what way. This enables us 
to obtain 
a clean and concise lower bound for the generalized hybrid search problem.

\paragraph{Hybrid Search Algorithms: Design and Analysis.}
In the second part of our work we construct a hybrid 
algorithm
for
the search problem for an arbitrary distribution $\distr$ and show that
in several interesting cases (e.g., Bernoulli) algorithm is optimal, 
hence
leading to tight query complexity in the hybrid model. Inspired by
our hardness analysis, our algorithm proceeds in a 
two-stage
fashion:
\begin{itemize}
\item The first stage is purely \emph{classical}. We query the
  $\cquery$ inputs $x$ that are the most likely to be assigned
  the value 1 under
  $\distr$. More precisely, for any $x$ in the input domain, let the function
  $\weight(x)=\sum_f D(f) \cdot f(x)$, which can be viewed as the
  (unnormalized) probability that $f(x) = 1$ with $f$ drawn from
  $\distr$. Let $S$ be the set of inputs whose $\weight(x)$ values are
  the $\cquery$-highest (ties are broken arbitrarily). 
  Then algorithm queries
  all the points $x\in S$. If none of them give a solution, we move on to the
  second stage.
\item The second stage is fully \emph{quantum}. We run a modified Grover's
  algorithm $\cA$ which is tailored to the prior knowledge on the
  distribution $\distr$. Instead of starting from an equal
  superposition of all points in the search space
  as in the standard Grover's search
  algorithm, we construct an initial state in which the amplitude of
  each point is proportional to $\weight(x)$. Then, for each of the $\qquery$ quantum queries, two reflection
  operators are applied to rotate the initial state towards a
  target state encoding the solutions. We give a comprehensive
  analysis and derive a clean lower bound for the success probability
  of $\cA$ on the distribution $\distr$, which amounts to
  $\qquery^2\cdot\frac{\sum_x \weight^2(x)}{\sum_x\weight(x)}$. To be
  more clear, 
  for the algorithm in the second stage, we 
  define an induced distribution $\tilde \distr$ by restricting and
  (re-normalizing) $\distr$ to functions $f$ satisfying $f(x)=0$ for
  all $x\in S$. Then we invoke the quantum algorithm $\cA$ on
  $\tilde \distr$ in a modular way.
\end{itemize}

Note that the hybrid algorithm needs to compute the values
$\weight(x)$ from the description of the target distribution $\distr$
and during the quantum procedure, the algorithm will implement a unitary dependent on the
$\weight(x)$ values, hence the algorithm needs not be time efficient.

We can show that the success probability of the hybrid algorithm is
at least the average of the success probabilities of the classical
stage and of the quantum stage. In some special cases, such as
the Bernoulli distribution, both the classical probability (i.e., at
least one success in $\cquery$ Bernoulli trials) and the weights
$w(x)$ (hence the quantum success probability) are easy to derive. We
can show that the hybrid algorithm gives matching lower bounds to the
hardness bounds that we proved in the first part of our work.



We believe that the hybrid query model is both of theoretical and
practical importance. Since near-term quantum computers are limited
and expensive, it is to the interest of a party to supplement it
with massive classical computational power. This also reflects the fact that those parties who have early access to quantum computers (e.g.,
big companies and government agencies) largely coincide with
those who are capable of employing classical clusters and
supercomputers. Next, we discuss 
some future directions.

One immediate question is to study other problems in the hybrid query
model. The work of~\cite{HLS22} also proves the hardness of the collision problem by their 
generalized recording technique in
the hybrid query model. It would be useful to further develop  techniques
and establish more query complexity results.

Our applications to hash functions and Bitcoin blockchains can be seen
as analyzing cryptographic constructions in the 
QRO
model against hybrid adversaries. Many block ciphers rely on 
a different
model, known as the ideal cipher model. 
As a simple
example, the Even-Mansour cipher encrypts by
$E_k: m \mapsto \sigma(k\oplus m)\oplus k$, where $\sigma$ is a random
permutation given as an oracle. As it turns out, this classically secure cipher is
completely broken when quantum queries are allowed to both $E_k$ and
$\sigma$~\cite{KM10}. Since the secret key $k$ is managed by honest
users, it is debatable whether superposition access to $E_k$ is
realistic. There has been progress in re-establishing the cipher's security under a
partially quantum adversary with quantum access to $\sigma$ but
classical access to $E_k$~\cite{JST21,ABKM22}. The hybrid query model
we consider in this work suggests further relaxing the queries to
$\sigma$ to be a hybrid of classical and quantum ones, and it would be
valuable to re-examine the security of such schemes in the ideal cipher
model.

Querying an oracle also appears more broadly in many other cryptographic scenarios. Security definitions often give some
algorithm as an oracle to the adversary, such as an
encryption oracle in the chosen-plaintext-attack (CPA) game and a
signing oracle in formalizing unforgeability of digital
signatures. There has been a considerable effort of settling
appropriate definitions and constructions (e.g., quantum-accessible
pseudorandom functions, encryption and signatures) when quantum
adversaries are granted superposition queries to these
oracles (cf.~\cite{BZ13,Zhandry15_ibe,AMRS20,Zhandry21_qprf,CEV23}). 
Extending such efforts to the hybrid-adversary landscape would
offer fine-grained security assessments of post-quantum cryptosystems.
Finally, in the context of complexity theory, the study of hybrid algorithms is further motivated by related models 
focusing on
the interplay between classical computation and near-future quantum devices \cite{CCHL22} and between circuit depth and quantum queries \cite{SZ19, CM20, CCL23}.

\paragraph{Updates from previous version (\url{https://eprint.iacr.org/archive/2023/798/20230703:172204}) of this paper}: {We have added the design and analysis of fully quantum as well as hybrid algorithms for distributional search, and have shown the optimality for special cases. This supplements the hardness results in the previous version.}

\subsection{Organization of the paper}
The rest of the paper is organized as follows. 
The generalized search problem, which we call {\em Distributional Search}, is defined in Section~\ref{sec:dsearch},
and its hybrid quantum-classical  
hardness is proven in Section~\ref{sec:proof} and Section~\ref{sec:progress}.
Section~\ref{sec:cases} describes two case studies---Grover-like Search and Bernoulli Search and 
Section~\ref{sec:apps} demonstrates applications of the Bernoulli Search results, where we show the generic security of hash functions (Sections~\ref{sec:hash}) and the security of the Bitcoin blockchain against hybrid adversaries (Sections~\ref{sec:pqbc_hybrid}). 
In Section~\ref{sec:quantum_algorithm} we construct a quantum search algorithm and establish its probability of success for any target distribution. 
Section~\ref{sec:hybrid_algorithm} describes our proposed hybrid search algorithm,
as well as the 
approach for the analysis of general distributions, and also shows the optimality of the algorithm for particular distributions.

\section{Distributional Search with Hybrid Strategies}
\label{sec:search} 
\subsection{The Distributional Search Problem}
\label{sec:dsearch} 

The underlying problem we consider is the search for a preimage of $1$ of an arbitrarily
distributed black-box boolean function.

\begin{mdframed}[linecolor=black!7, backgroundcolor=black!7]
  \begin{center}
    \textbf{Distributional Search Problem ($\dsearch$) }
  \end{center}

  Let $\distr$ be an arbitrary distribution supported on the function
  family $\ff : = \ffunc$.

\noindent\textbf{Given}: Black-box access to function $f$ drawn
from distribution $D$. 

\noindent\textbf{Goal}: Find $x$ 
such that $f(x) = 1$ if there exists such
an $x$.
\end{mdframed}
It is not surprising that the hardness of the problem is crucially influenced by the
number of preimages of $1$ \emph{on average} under $\distr$; however, what is
interesting about our study is that we can show a clean quantitative relation. Let
$f : X\to \bits$ be an arbitrary function. Define the projector on the
space spanned by preimages 
of $1$:
\[ \pi_f \defeq \sum_{x: f(x) = 1} | x\rangle\langle x| \, . \]
Denote
$\pi_f^\perp \defeq \identity - \pi_f$,
and let $\distr$ be a distribution on
$\ff$. We define the value that captures the \emph{average}
fraction of preimages of $1$ as:

\begin{definition}[\solnl] The average fraction of preimages of $1$ is defined as:
    \begin{equation}
        \solnl \defeq \sup_{\varphi : \norm{\varphi} \leq 1} \left(\expt_{f\gets \distr} \norm{\pi_f \varphi}^2  \right)
    \end{equation}    
    where $\norm{ \varphi }$ denotes the Euclidean norm of the quantum state $\varphi$.
\end{definition}

In this paper, we are able to establish the following bound for the 
success probability 
of solving $\dsearch$, which constitutes one of our main results:
\begin{theorem}[Hardness of $\dsearch$]
\label{thm:main} For any
algorithm $\adv$ making up to $\cquery$ classical queries and
$\qquery$ quantum queries, it holds that:
\[
  \succp_{\adv,\distr}: = \Pr_{f\gets \distr}[ f(x)=1: x\gets
  \adv^{{f}} ] \le \solnl\cdot (2\sqrt{\cquery} + 2\qquery + 1)^2 \, .
\]
\label{thm:dsearch}
\end{theorem}
\vspace{-1em}

Next, we turn to proving the above result.


\subsection{Hardness of \texorpdfstring{$\dsearch$}{}}
\label{sec:proof} 


\subsubsection{Preliminaries and Overview}
\label{sec:notation} 

We first formally describe an oracle function for the case of quantum
and pseudo-classical queries.

\begin{definition}[Query Operators] \label{def:query_operators} We
  define the following operators, which describe the actions of
  quantum and pseudo-classical oracles for a hybrid algorithm given a boolean
  function $f$.
  \begin{itemize}
  \item A pseudo-classical oracle is described by
    \begin{equation*}
      \coracle_{f, b} \defeq \sum_{x: f(x)=b} \ket{x}\bra{x} \otimes \identity \otimes \ket{b}
    \end{equation*}
  \item A quantum oracle is described by
    \begin{equation*}
      \qoracle_f \defeq \sum_{x, b} |x\rangle \langle x| \otimes \identity \otimes \ket{b \oplus f(x)} \langle b |
    \end{equation*}      
  \end{itemize}
\end{definition}

We denote $\Pi_f \defeq \pi_f\otimes \identity$ ($\identity$ operates on the output and ancilla registers) and $\Pi_f^\perp \defeq \identity - \Pi_f$ ($\identity$ operates on the entire system). Then on a pseudo-classical query, the two operators
$\coracle_{f,0} = \Pi_f^\perp \otimes \ket{0}$ and
$\coracle_{f,1} = \Pi_f \otimes \ket{1}$ correspond to the two
possible measurement outcomes. It is more convenient to answer
quantum queries by the corresponding phase oracle:
\[ \qoracle_f \defeq \identity - 2 \Pi_f \, . \] 
This can be seen as
setting the output register of the standard oracle in $\ket{-}$, and
as a result, a quantum query flips the signs of the $1$-preimages. 

When running a hybrid query algorithm with $f$, we will keep track of
the (sub-normalized) pure state $\psi_f^{(t)}$, which denotes the state of the algorithm on input $f$ after $t$ queries in the situation where
every pseudo-classical query measures $0$ (we will call this the
$0$-branch of $\adv^f$). Namely, consider an arbitrary algorithm with at
most $\aquery$ queries ($\qquery$ quantum and $\cquery$
pseudo-classical) specified by a sequence of unitary
operators\footnote{Dimensions may grow depending on the arrangement of
 the pseudo-classical queries.}
$(U^{(0)},U^{(1)},\ldots, U^{(\aquery)})$. Let
$T_c = \{t: \text{$t$-th query is pseudo-classical}\}$ and
$T_q = \{ t: \text{$t$-th query is quantum}\}$. Then $\psi_f^{(t)}$
is defined recursively by
\begin{equation} \label{eq:def_state_evolution}
\psi_f^{(t)} \defeq \begin{cases}
                     U^{(t)}\coracle_{f,0} \psi_f^{(t-1)}, & \text{if }  t\in T_c \, ;\\
                     U^{(t)}\qoracle_f \psi_f^{(t-1)} & \text{if }
                                                        t\in T_q \, .
                  \end{cases}     
\end{equation}

From this definition, the projection of $\psi_f^{(t)}$ under $\Pi_f^\perp$
characterizes the event that an algorithm fails to find a $1$-preimage. 

\begin{lemma} For any algorithm $\adv$, the failure probability of
  finding a $1$-preimage of $f$ after $t$ queries is
  \[ \pfail_f^{(t)} =\Pr[f(x) \ne 1 : x\gets \adv^f] \ge
    \norm{\Pi_f^{\perp} \psi_f^{(t)}}^2\, . \] Hence, the failure
  probability with respect to distribution $\distr$ satisfies
  \[ \pfail_\distr^{(t)} = \expt_{f\gets D} \pfail_f^{(t)}\ge
    \expt_{f\gets \distr}\norm{ \Pi_f^{\perp} \psi_f^{(t)}}^2 \, .\]
  \label{lemma:failure prob}
\end{lemma}

\vspace{-.1in}
Thus, our goal 
becomes lower-bounding
$\norm{\Pi_f^{\perp} \psi_f^{(t)}}$. To do this, we consider running
the same algorithm, but with a null function:
\[ \fnull: x\mapsto 0,\forall x\in X \, .\] In this case, a quantum
query is equivalent to applying identity (denoted
$Q_\emptyset \defeq \identity$), and a pseudo-classical query does not
tamper the input state either, but just appends $\ket{0}$. To be
precise, we define
\[\coracle_{\emptyset,0}\defeq \identity \otimes \ket{0} \, ,\]
and at each step $t\ge 0$, the state of the algorithm denoted by $\phi^{(t)}$ can be described as:
\begin{equation*}
  \phi^{(t)} =
  \begin{cases}
    U^{(t)}\coracle_{\emptyset,0} \phi^{(t-1)}, & \text{if }  t\in T_c \, ;\\
    U^{(t)} \phi^{(t-1)} & \text{if } t\in T_q \, .
  \end{cases}
\end{equation*}
              
Without loss of generality we assume initially
$\psi_f^{(0)} = \phi^{(0)} = \ket{0}$, and hence
$\norm{\Pi_f^\perp \psi_f^{(0)}} = \norm{\Pi_f^\perp\phi^{(0)}} =
1$. In order to succeed, algorithm $\adv^f$ needs to move
$\psi_f^{(t)}$ away from the kernel of $\Pi_f^\perp$ or reduce its
norm. This motivates defining the progress measures below.

\begin{definition}[Progress Measures $(A^{(t)}, B^{(t)})$] \label{def:progress_measures}
For any
  function $f$ and $t \ge 0$, define
  \begin{equation*}
    A_f^{(t)} \defeq  \left|\langle \phi^{(t)},
      \psi_f^{(t)}\rangle\right|^2 \, , \quad 
    B_f^{(t)} \defeq  \left\| \psi_f^{(t)} \right\|^2 - \left|\langle
      \phi^{(t)} , \psi_f^{(t)}\rangle \right|^2 \, .
  \end{equation*}
  Given a distribution $\distr$ on $\ff$, define the expected progress
  measures by
  \begin{equation*}
    A_\distr^{(t)} \defeq \expt_{f\gets \distr} \left( A_f^{(t)} \right)\, , \quad 
    B_\distr^{(t)} \defeq \expt_{f\gets \distr} \left( B_f^{(t)} \right) \, .
  \end{equation*} 
\label{def:2progress}
\end{definition}
\vspace{-1em}
Notice that
\begin{equation*}
  A_f^{(t)} + B_f^{(t)} = \norm{\psi_f^{(t)}}^2 , \quad A_f^{(0)} = 1
  , \quad B_f^{(0)} = 0 \, .
\end{equation*}

We will show that $A_\distr^{(t)} - B_\distr^{(t)}$ essentially lower bounds the
failure probability $\pfail_\distr^{(t)}$ (\cref{lemma:pfail}). Hence, an
algorithm's objective would be to \emph{reduce} $A_\distr^{(t)}$ and
\emph{increase} $B_\distr^{(t)}$. However, we can limit how much change can
occur after $\aquery$ queries (\cref{prop:progressbound}). This is by
carefully analyzing the effect of each quantum or pseudo-classical query
(\cref{lemma:progress_fixed,lemma:progress_general}). Roughly
speaking,

\begin{tiret}
\item A quantum query reduces $A_\distr^{(t)}$ by at most
  $4 \sqrt{\solnl \cdot B_\distr^{(t)}}$ and increases $B_\distr^{(t)}$ by the same
  amount (as a quantum query does not affect $\norm{\psi_f^{(t)}}^2$), and
\item A pseudo-classical query increases $B_\distr^{(t)}$ by at most
  $\solnl$, while a part $z^{(t)}$ of $B_\distr^{(t)}$ can also be spent to
  decrease $A_\distr^{(t)}$ by $\sqrt{\solnl \cdot z^{(t)}}$.
\end{tiret}


\begin{table}[h!]
\centering
\begin{tabular}{c | c } 
  \hline\hline
  $\pi_f$  & $\sum_{x:f(x) = 1} \ket{x}\bra{x}$  \\
  \hline
  $\Pi_f$  & $\pi_f\otimes \identity$ ($\identity$ on ancilla registers) \\
  \hline
  $\pfail_f$ & $\Pr[f(x) \ne 1 : x\gets \adv^f]$ \text{(Failure
               probability with $f$)} \\
  \hline
  $\pfail_\distr$ & $\expt_\distr \pfail_f$ \text{(Failure probability
                    with $f\gets \distr$)} \\
  \hline
  $\phi^{(0)} = \psi^{(0)}$ & \text{Initial state} \\
  \hline
  $\phi^{(t)}$ & \text{State after $t$-th
                 query in $\adv^{\fnull}$} \\
  \hline
  $\psi_f^{(t)}$ & \text{State on the 0-branch after $t$-th
                   query in $\adv^{f}$} \\
  \hline
  $\qoracle_{f}$   & $\identity - 2\Pi_f$ \text(quantum
                     oracle of $f$)\\
  \hline
  $\qoracle_{\emptyset}$   & $\identity$ \text(quantum
                             oracle of $\fnull$)\\
  \hline
  $\coracle_{f,0}$  & $\Pi_f^\perp \otimes \ket{0}$ \text(pseudo-classical
                      oracle of $f$)\\
  \hline
  $\coracle_{f,1}$   & $\Pi_f \otimes \ket{1}$
                       \text(pseudo-classical oracle of $f$)\\
  \hline
  $\coracle_{\emptyset,0}$  & $\identity \otimes \ket{0}$ \text(pseudo-classical
                              oracle of $\fnull$)\\
  \hline
  $\tproj_f^{(t)}$ & $\norm{\Pi_f \phi^{(t)}}^2$ \\
  \hline
  $\tproj^{(t)}$ & $\expt_\distr (\tproj_f^{(t)})$ \\
  \hline\hline  
\end{tabular}
\caption{Summary of variables and quantities used in our $\dsearch$
  analysis.}
\label{table:notations}
\end{table}

\subsubsection{Proof of Theorem~\ref{thm:main}}
\label{sec:mainproof} 

First off, we state the Cauchy-Schwarz inequality for random variables and derive a lemma that is useful in several places. 

\begin{lemma}[Cauchy-Schwarz] For any random variables $X$, $Y$, it
  holds that:
\begin{equation*}
    \left| \expt{[X Y]} \right|^2 \leq \expt{[X^2]}\cdot \expt{[Y^2]} .
\end{equation*}
\label{lemma:cs}
\end{lemma}

\vspace{-.3in} 

As a direct consequence of the Cauchy-Schwarz inequality, we have the following useful corollary.
\begin{corollary}
Let $Z$ be a discrete random variable, and
$g(Z)$ and $h(Z)$ be two
non-negative functions. Then it holds that: 
\begin{equation*} \expt_{Z} \left(\sqrt{g(Z)\cdot h(Z)}\right) \le \sqrt{\expt_Z g(Z)
    \cdot \expt_Zh(Z)} \, .
\end{equation*}
\label{lemma:expproduct}
\end{corollary}

\vspace{-.3in} 

It will be helpful to consider a
two-dimensional plane in our analysis, which we now define explicitly.

\begin{definition}[Useful 2-D Plane] \label{def:2dplane} For $t\ge 0$,
  let
  \begin{equation*}
    \phi_f^{(t)} \defeq \frac{\Pi_f \phi^{(t)}}{\norm{\Pi_f \phi^{(t)}}}
    = \Pi_f\phi^{(t)} / \sqrt{\tproj_f^{(t)}} \, , \qquad
    \phi_f^{(t)\perp} \defeq \frac{\Pi_f^\perp
      \phi^{(t)}}{\norm{\Pi_f^\perp \phi^{(t)}}} =
    \Pi_f^{\perp}\phi^{(t)} / \sqrt{1 - \tproj_f^{(t)}} \, 
  \end{equation*} be the
  \emph{normalized} vectors resulting of projecting $\phi$ on the orthogonal
  subspaces spanned by $1$ and $0$ preimages of $f$, respectively,
  and let $\Phi^{(t)}$ 
  be the $2$-dimensional plane spanned by
  $\{\phi_f^{(t)},\phi_f^{(t)\perp}\}$. Then $\phi^{(t)\perp}$ is
  identified as the normalized state perpendicular to $\phi^{(t)}$ in
  $\Phi^{(t)}$, i.e.,
  \begin{equation*}
    \phi^{(t)\perp} \defeq \phi_f^{(t)} \sqrt{1-\tproj_f^{(t)}} - \phi_f^{(t)\perp}
    \sqrt{\tproj_f^{(t)}} \, .
  \end{equation*}  
\end{definition}  

It is 
useful to decompose $\psi_f^{(t)}$ with respect to
$\Phi^{(t)}$:

\begin{lemma}[Decomposition of $\psi_f^{(t)}$ wrt $\Phi^{(t)}$]
  Let $a$ and $b$ be projecting $\psi_f^{(t)}$ on the plane
  $\Phi^{(t)}$ and then decomposing it under basis
  $\{\phi^{(t)}, \phi^{(t)\perp}\}$, and let $c$ be the remaining component of
  $\psi_f^{(t)}$ orthogonal to $\Phi^{(t)}$, i.e.,
  $c \perp \Phi^{(t)}$. Then $\psi_f^{(t)}$ can be expressed as
  $\psi_f^{(t)} = a + b + c$ with
  \begin{equation*}
    a = \phi^{(t)} \sqrt{A_f^{(t)}}\, , \qquad b = \omega
    \sqrt{B_f^{(t)} -\norm{c}^2}\cdot \phi^{(t)\perp}\, , 
  \end{equation*}
  where $\omega$ is a complex phase (i.e., $|\omega| = 1$) of the
  vector
  $\psi_f^{(t)} - \langle \psi_f^{(t)}, \phi_f^{(t)}\rangle \cdot
  \phi^{(t)} - c$. As a result,
  \begin{equation*}
    \Pi_f^\perp\psi_f^{(t)} =
    \phi_f^{(t)\perp}\left(\sqrt{1 -
        \tproj_f^{(t)}}\sqrt{A_f^{(t)}} - \sqrt{\tproj_f^{(t)}}
      \cdot \omega \sqrt{B_f^{(t)}
        -\norm{c}^2} \right) + c_f^\perp \, ,
  \end{equation*}
  with $c_f^\perp := \Pi_f^{\perp}c$.
  \label{lemma:psidecompose}
\end{lemma}

Intuitively, for the next result, the goal is to relate the failure probability with the progress measures $A$ and $B$. To do so, we will first relate the failure probability with the norm of the non-solution component. By decomposing this norm in terms of the two progress measures A and B and an orthogonal component which can be removed, we can determine a lower bound on the failure probability as a function of the two progress measure after each performed query.

\begin{lemma} For any fixed $f$ and $t\ge 0$,
  \[ \pfail_f^{(t)} \ge A_f^{(t)} - \gamma_f^{(t)} -
    2\sqrt{\gamma_f^{(t)}\cdot B_f^{(t)}} \, . \]
  \label{lemma:pfail_fixed}
\end{lemma}

\begin{proof} For convenience, we omit writing the superscript $(t)$ 
in this proof. We first show that
  $\norm{\pi_f^\perp \psi_f }\ge \sqrt{(1 - \gamma_f)A_f} -
  \sqrt{\gamma_f B_f}$. By~\cref{lemma:psidecompose}, we have that
 \begin{equation*}
    \Pi_f^\perp\psi_f =
    {\phi_f^\perp}\left(\sqrt{1 -
        \tproj_f}\sqrt{A_f} - \sqrt{\tproj_f}
      \cdot \omega \sqrt{B_f
        -\norm{c}^2} \right) + c_f^\perp \, ,
  \end{equation*}
  with $c_f^\perp := \pi_f^{\perp}c$. Since $c \perp \Phi$, it follows
  that
  \[ \langle \phi_f^\perp, c_f^\perp \rangle = \langle \phi_f^\perp ,
    \Pi_f^\perp c \rangle = \langle \Pi_f^\perp \phi_f^\perp , c
    \rangle = \langle \phi_f^\perp , c \rangle = 0 \, .\] We can then
  obtain: 
  \begin{equation*}
    \norm{\Pi_f^\perp \psi_f} = \left|\sqrt{1 -
        \tproj_f}\cdot \sqrt{A_f} - \sqrt{\tproj_f}
      \cdot \omega \sqrt{B_f
        -\norm{c}^2}\right| + \norm{c_f^\perp}
  \end{equation*}

  Hence by choosing $c = 0, \omega=1$, we get that
  \[ \norm{\Pi_f^\perp \psi_f} \ge \sqrt{(1 - \tproj_f)A_f} -
    \sqrt{\tproj_f B_f} \, .\]
  
  Therefore we can lower bound the failure probability
  \begin{align*}
    \pfail_f & \ge \norm{\pi_f^\perp \psi_f}^2 \\
             & \ge (1 - \tproj_f)A_f - 2\sqrt{(1-\tproj_f)\tproj_f B_f}\\
             & \ge A_f - \tproj_f - 2\sqrt{\tproj_f B_f} \qquad \text{($A_f,\tproj_f \le 1$)}
  \end{align*}
\end{proof}

Taking the expectation over $\distr$, we can express the failure
probability with respect to the distribution.

\begin{lemma} For any distribution $\distr$ and $t\ge 0$,
  \[ \pfail_\distr^{(t)} \ge A^{(t)} - \gamma^{(t)} -
    2\sqrt{\gamma^{(t)}\cdot B^{(t)}} \, . \]
  \label{lemma:pfail}
\end{lemma}
\begin{proof}
  \begin{align*}
    \pfail_\distr^{(t)} & = \expt_{f\gets \distr} (\pfail_f^{(t)}) \\
                        & \ge \expt_\distr(A_f^{(t)}) - \expt_\distr(\gamma_f^{(t)}) -
                          2\expt_\distr \sqrt{\gamma_f^{(t)}\cdot B_f^{(t)}}   \qquad
                          \text{(Linearity of expectation)}\\
                        & \ge  A^{(t)} - \gamma^{(t)} -
                          2\sqrt{\expt_\distr(\gamma_f^{(t)})\cdot \expt_\distr(B_f^{(t)})} \qquad
                          \text{(\cref{lemma:expproduct})}\\
                        & =  A^{(t)} - \gamma^{(t)} - 2\sqrt{\gamma^{(t)}\cdot B^{(t)}} 
  \end{align*} 
\end{proof}

We can also relate $\gamma^{(t)}$ to the value $\solnl$ determined by the
distribution $\distr$:

\begin{lemma} For any $t\ge 0$ and any distribution $\distr$, we have:
  $\gamma^{(t)} \le \solnl$.
  \label{lemma:gammavssolnl}
\end{lemma}
\begin{proof}
\begin{equation*}
    \begin{split}
    \gamma^{(t)} &:= \expt_{f\gets \distr} \norm{ \Pi_f \phi^{(t)}}^2 = \expt_{f\gets \distr} \norm{ (\pi_f \otimes \identity) \phi^{(t)} }^2 \\
    \end{split}
\end{equation*}
We write $\phi^{(t)} = \sum_i \alpha_i \ket{u_i} \otimes \ket{v_i}$ under the Schmidt decomposition, where $\alpha_i \geq 0$ such that $\sum_i \alpha_i^2 = 1$ are the Schmidt coefficients, and $\{\ket{u_i}\}$
are orthonormal states on the system of the input register and 
$\{\ket{v_i}\}$
are orthonormal states on the system of output and ancilla registers.

Then we can rewrite $\gamma^{(t)} $ as
\begin{equation*}
    \begin{split}
        \gamma^{(t)} &:= \expt_{f\gets \distr} \norm{(\pi_f \otimes \identity) \phi^{(t)} }^2 = \expt_{f\gets \distr} \norm{(\pi_f \otimes \identity) \left(  \sum_i \alpha_i \ket{u_i} \otimes \ket{v_i} \right)}^2 \\
        &= \expt_{f\gets \distr} \norm{ \sum_i \alpha_i (\pi_f \ket{u_i}) \otimes \ket{v_i}}^2   \\
        &= \expt_{f\gets \distr} \sum_i \alpha_i^2 \norm{(\pi_f \ket{u_i}) \otimes \ket{v_i}}^2 \qquad \text{($\ket{v_i}$ are orthogonal)} \\
        &= \expt_{f\gets \distr} \sum_i \alpha_i^2 \norm{\pi_f \ket{u_i}}^2 \cdot \norm{\ket{v_i}}^2 \qquad \text{($\norm{a \otimes b} = \norm{a} \cdot \norm{b}$)} \\
         &= \expt_{f\gets \distr} \sum_i \alpha_i^2 \norm{\pi_f \ket{u_i}}^2 \\ 
         &= \sum_i \alpha_i^2 \expt_{f\gets \distr}  \norm{\pi_f \ket{u_i}}^2 \\
         & \leq \sum_i \alpha_i^2 \solnl \qquad \text{(definition of $\solnl$)} \\
         &= \solnl \sum_i \alpha_i^2
         = \solnl
    \end{split}
\end{equation*}
\end{proof}
\begin{proposition}[Bounding Progress Measures]
\label{prop:bounding} After
  $\aquery = \cquery + \qquery$ queries,
  \[ A^{(\aquery)} \ge 1 - 4 \solnl \cdot (\sqrt{\cquery} + \qquery)^2
    \, , \quad B^{(\aquery)} \le \solnl \cdot (\sqrt{\cquery} +
    2\qquery)^2 \, .\]
  \label{prop:progressbound}
\end{proposition}
Proving~\cref{prop:progressbound} is the most involved step technically speaking. We
present the details separately in~\cref{sec:progress} and here we apply it to prove~\cref{thm:dsearch}.

\begin{proof}[Proof of~\cref{thm:dsearch}]

 Assuming the bounds above on the two progress measures, we obtain that:
  \begin{align*}    
    \pfail^{(\aquery)} & \ge  1 - 4\gamma^{(\tau)} \cdot (\sqrt{\cquery}
                         + \qquery)^2 - \gamma^{(\tau)} - 2 \gamma^{(t)}\cdot
                         (\sqrt{\cquery} + 2 \qquery)
                         \qquad \text{(\cref{prop:progressbound})} \\
                       & = 1 - \gamma^{(\tau)} \cdot (4(\sqrt{\cquery}+\qquery) + 2\sqrt{\cquery}
                         + 4\qquery + 1)\\
                       & \ge 1 - \gamma^{(\tau)} \cdot (2(\sqrt{\cquery}+\qquery) + 1)^2 \qquad \text{($\cquery \ge 0$)} \\
                       & \ge 1 - \solnl\cdot
                         (2\sqrt{\cquery} + 2\qquery + 1)^2 \qquad 
                         \text{($\gamma^{(\aquery)} \le \solnl$~\cref{lemma:gammavssolnl})}
  \end{align*}
  Therefore,
  \[\succp_{\adv,\distr} \le 1 - \pfail^{(\aquery)} \le  \solnl\cdot
    (2\sqrt{\cquery} + 2\qquery + 1)^2\, .\]
\end{proof}





\subsection{Case Studies} 
\label{sec:cases} 

In this section, we will apply our main result to 
two common
function
distributions. As a common ingredient, it will be helpful to consider the following
indicator random variable:
\begin{equation*}
\identity_{x}^f \defeq
  \begin{cases}
    1 & \text{if }  f(x) = 1 \, ;\\
    0 & \text{if }
        f(x) = 0 \, ,    
  \end{cases}
\end{equation*}
for all $f\in \func$ and $x\in X$. Then, for a distribution $\distr$,
\begin{equation*}
  \expt_{f\gets \distr} (\identity_x^f) = \Pr_{f\gets \distr}[f(x) = 1]
  \, .
\end{equation*}

\subsubsection{Grover-like Search}
\label{sec:grover} 

The first interesting case is a general Grover-type search. We
consider a distribution $\distr_w$ which is \emph{uniform} over
functions that exactly map $w$ inputs to $1$. In other words, drawing
$f\gets \distr_w$ is equivalent to sampling a subset $S\subseteq X$
with $|S| = w$ uniformly at random and set $f(x) = 1$ 
if and only if
$x\in S$. We consider the resulting multi-uniform search problem:

\begin{mdframed}[linecolor=black!7, backgroundcolor=black!7]
  \begin{center}
    \textbf{Multi-Uniform Search}
  \end{center}
  \noindent\textbf{Given}: $f \gets \distr_w$, which maps a uniform
  size-$w$ subset to 1.

  \noindent\textbf{Goal}: Find $x$ such that $f(x) = 1$.
\end{mdframed}

\begin{theorem}
  For any adversary $\adv$ making up to $\cquery$ classical queries
  and $\qquery$ quantum queries,

  \[ \succp_{\adv,\distr_w} \le \frac{w}{M} \cdot (2\sqrt{\cquery} + 2
    \qquery + 1)^2 \, ,\] where $M = |X|$ is the domain size.
\label{thm:wunifsearch}
\end{theorem}

\begin{proof}
  We just need to show that
  $\solnl = \sup_{\varphi : \norm{\varphi} \leq 1} \expt_{f\gets \distr_w} (\norm{\pi_f \varphi}^2)\le \frac{w}{M}$ in
  this case. Consider an arbitrary unit vector
  $\varphi = \sum_{x} \alpha_x\ket{x}$ with $\sum_x |\alpha_x|^2 = 1$.
  \begin{align*}
    \expt_{f\gets \distr_w} (\norm{\pi_f\varphi}^2) & =
                                                    \expt_{f\gets \distr_w}
                                                    (\left|\sum_x
                                                    \alpha_x
                                                    \identity_x^f
                                                    \ket{x} \right|^2)\\
                                                  & = \sum_x
                                                    |\alpha_x|^2
                                                    \cdot 
                                                    \expt_{f\gets
                                                    \distr_w}(\identity_x^f
                                                    )\\
                                                  & = \sum_x
                                                    |\alpha_x|^2
                                                    \cdot \Pr_{f\gets \distr_w}[ f(x) = 1] \\
                                                  & = \frac{w}{M} \, .
  \end{align*}
\end{proof}

We note two special scenarios. When $w =1$, this 
reproduces
Rosmanis's result~\cite{Ros22}, and when $\cquery = 0$, 
our result
reproduces the fully quantum query complexity of Grover search with
multiple marked items (cf.~\cite{BBBV97,Zalka99}). 

\subsubsection{Bernoulli Search}
\label{sec:bersearch} 

The second
interesting case 
is what we call a Bernoulli distribution
$\distr_\berpar$ on $\ff$, as specified below:

\begin{mdframed}[linecolor=black!7, backgroundcolor=black!7]

  \begin{center}
    \textbf{Bernoulli Search}
  \end{center}
  \noindent\textbf{Given}: $f \gets \distr_\berpar$ drawn via the following
  sampling procedure:

  For each $x\in X$, \emph{independently} set

  \[ f(x) = \left\{ \begin{matrix}
                      1, & \text{with probability } \bpara \\
                      0, & \text{otherwise} \\
                    \end{matrix}\right. \, .\]

\noindent\textbf{Goal}: Find $x$ such that $f(x) = 1$.
\end{mdframed}

\begin{theorem} \label{thm:bernoulli}
  For any adversary $\adv$ making up to $\cquery$ classical queries
  and $\qquery$ quantum queries,
  \[ \succp_{\adv,\distr_\eta} \le \berpar \cdot (2\sqrt{\cquery} + 2
    \qquery + 1)^2 \, .\]
\label{thm:bersearch}
\end{theorem}

\begin{proof}
  Consider an arbitrary unit vector
  $\varphi = \sum_{x} \alpha_x\ket{x}$ with $\sum_x |\alpha_x|^2 =
  1$. Again, we just need to show that
  $\expt_{f\gets \distr_\berpar} (\norm{\pi_f \varphi}^2 ) \le
  \berpar$. 
  Similarly as above, 
  \begin{equation*}
    \expt_{f\gets \distr_\berpar} (\norm{\pi_f\varphi}^2 )= \sum_x |\alpha_x|^2
    \cdot \Pr_{f\gets \distr_\berpar}[ f(x) = 1]  = \berpar \, .    
  \end{equation*}
\end{proof}

Note that when $\cquery = 0$, this bound reproduces the complexity of
Bernoulli Search using fully quantum queries 
(cf.~\cite{HRS16,ARU14}).

\subsection{Bounding the Progress Measures
  (Proposition~\ref{prop:bounding})} %
\label{sec:progress}

We repeat the proposition statement for convenience here:

\noindent {\bf Proposition~\ref{prop:bounding} (Bounding the Progress Measures).}
{\em After
  $\aquery = \cquery + \qquery$ queries,
  \[ A^{(\aquery)} \ge 1 - 4 \solnl \cdot (\sqrt{\cquery} + \qquery)^2
    \, , \quad B^{(\aquery)} \le \solnl \cdot (\sqrt{\cquery} +
    2\qquery)^2 \, .\] } 

First, we consider a fixed function $f$, and bound how much each query
can possibly reduce $A_f^{(t)}$ and increase $B_f^{(t)}$.
\begin{lemma}[Progress Measures for a Fixed
  Function] \label{lemma:progress_fixed} For every $t$
  the progress measures after the $t + 1$-th query satisfy the
  following recurrent relations:
\begin{itemize}
\item If the $t + 1$-th query is \emph{pseudo-classical}, then there
  exists a sequence $\left(z_f^{(t)}\right)_{t \geq 0}$, satisfying
  $0 \leq z_f^{t} \leq B_f^{(t)}$, such that:
  \begin{equation}
      \begin{split}
        A_f^{(t + 1)} &\geq A_f^{(t)} - 2 \gamma_f^{(t)} - 2 \cdot \sqrt{z_f^{(t)}} \cdot \sqrt{\gamma_f^{(t)}} \\
        B_f^{(t + 1)} &\leq B_f^{(t)} + \gamma_f^{(t)} - z_f^{(t)}
      \end{split}
      \end{equation}
    \item If the $t + 1$-th query is \emph{quantum}, then:
        \begin{equation}
            \begin{split}
                A_f^{(t + 1)} &\geq A_f^{(t)} - 4 \gamma_f^{(t)} - 4 \cdot \sqrt{B_f^{(t)}} \cdot  \sqrt{\gamma_f^{(t)}} \\
               B_f^{(t + 1)} &\leq B_f^{(t)} +  4 \gamma_f^{(t)} + 4 \cdot \sqrt{B_f^{(t)}} \cdot \sqrt{\gamma_f^{(t)}}
            \end{split}
    \end{equation}
\end{itemize}
\end{lemma}

\begin{proof} We analyze the two cases separately.
  
\paragraph{Pseudo-classical query case.}
If the $(t+1)$-th query is pseudo-classical, according to the evolution
of the state definition (\cref{eq:def_state_evolution}), the states
after the $t + 1$ query are:
\begin{equation*}
  \psi_f^{(t + 1)} = U^{(t + 1)} P_{f, 0} \psi_{f}^{(t)}\, , \qquad   
  \phi^{(t + 1)} = U^{(t + 1)} P_{\emptyset, 0} \phi^{(t)}\, . 
\end{equation*}

Therefore, we have: 
\begin{align*}
  A_f^{(t + 1)} & = |\langle \phi^{(t + 1)}, \psi_f^{(t + 1)}\rangle|^2
  \\
                & = |\langle P_{\emptyset, 0}\phi^{(t)}, P_{f, 0}\psi_f^{(t)}\rangle|^2 \\
                & = |\langle \phi^{(t)} \ket{0},
                  \Pi_f^{\perp}\psi_f^{(t)} \ket{0}\rangle|^2 \\
                & = |\langle \phi^{(t)},
                  \Pi_f^{\perp}\psi_f^{(t)}\rangle|^2
\end{align*}

By the decomposition of $\psi_{f}^{(t)}$ (\cref{lemma:psidecompose}), we
know that:

\begin{equation*}
  \Pi_f^\perp\psi_f^{(t)} =
  {\phi_f^{(t)\perp}} \left(\sqrt{1 -
      \tproj_f^{(t)}} \sqrt{A_f^{(t)}} - \sqrt{\tproj_f^{(t)}}
    \cdot \omega \sqrt{B_f^{(t)}
      -\norm{c}^2} \right) + c_f^\perp \, ,
\end{equation*}
  
where $|\omega| = 1$ and $c_f^\perp := \Pi_f^{\perp}c$ and
$c \perp \Phi = \text{span}\{\phi^{(t)}, \phi^{(t)\perp}\}$.

Note that
$\langle \phi^{(t)}, c_f^{\perp} \rangle = \langle \phi_f^{(t)\perp},
c_f^\perp \rangle = 0$ and
$\langle \phi^{(t)}, \phi_f^{(t)\perp} \rangle = \sqrt{ 1 -
  \tproj_f^{(t)}}$. As a result, we obtain:
\begin{align*}
  |\langle \phi^{(t)},
  \Pi_f^{\perp}\psi_f^{(t)}\rangle|^2 & = \left| \sqrt{1 -
                                        \tproj_f^{(t)}}\sqrt{A_f^{(t)}} - \sqrt{\tproj_f^{(t)}}
                                        \cdot \omega \sqrt{B_f^{(t)}
                                        -\norm{c}^2} \right |^2
                                        \cdot (1 - \tproj_f^{(t)})\\ 
                                      & \ge (1 - \tproj_f^{(t)})^2 A_f^{(t)} - 2\cdot \sqrt{\tproj_f^{(t)}}
                                        \cdot \sqrt{B_f^{(t)}
                                        -\norm{c}^2} \\
                                      & \ge A_f^{(t)} -
                                        2\tproj_f^{(t)} 
                                        - 2\cdot \sqrt{\tproj_f^{(t)}}
                                        \cdot \sqrt{B_f^{(t)}
                                        -\norm{c}^2} \, .
\end{align*}

Noting that $|| c_f^{\perp} ||^2  = || \Pi_f^{\perp} c ||^2 = ||c - \Pi_f c||^2 \leq ||c||^2$, we get:
\begin{equation*}
     A_f^{(t+1)} \ge A_f^{(t)} -  2\tproj_f^{(t)}  - 2\cdot \sqrt{\tproj_f^{(t)}}  \cdot \sqrt{B_f^{(t)} - \norm{c_f^{\perp}}^2}
\end{equation*}

Hence, by setting
$z_f^{(t)} \defeq B_f^{(t)} - \norm{c_f^{\perp}}^2 \in [0,B_f^{(t)}]$, we have that:

\begin{equation*}
  A_f^{(t+1)} \ge A_f^{(t)} - 2\tproj_f^{(t)} - 2 \sqrt{z_f^{(t)}} \cdot \sqrt{\tproj_f^{(t)}}  \, .
\end{equation*}

Next we analyze $B_f^{(t + 1)}$. By definition,

\begin{align*}
  B_f^{(t + 1)} & = \norm{\psi_f^{(t + 1)}}^2 - A_f^{(t + 1)} \\
                & = \norm{U^{(t + 1)} \coracle_{f, 0} \psi_f^{(t)}}^2
                  - A_f^{(t + 1)}\\
                & = \norm{\Pi_f^\perp \psi_f^{(t)}}^2 - A_f^{(t +
                  1)} \, .
\end{align*}

Denote
$E^{(t)} \defeq \left| \sqrt{1 - \tproj_f^{(t)}}\sqrt{A_f^{(t)}} - \sqrt{\tproj_f^{(t)}}
  \cdot \omega \sqrt{B_f^{(t)} -\norm{c}^2} \right|^2$.  Observe that:

\begin{equation*}
  \norm{\Pi_f^\perp \psi_f^{(t)}}^2  =  \norm{\phi_f^{\perp}}^2 \cdot E^{(t)} + \norm{c_f^{\perp}}^2 
  = E^{(t)} + \norm{c_f^{\perp}}^2  \, .
\end{equation*}

Meanwhile, from above,
\begin{equation*}  
  A_f^{(t+1)} = |\langle \phi^{(t)},
  \Pi_f^{\perp}\psi_f^{(t)}\rangle|^2 = E^{(t)}(1 - \tproj_f^{(t)}) \, .
\end{equation*}

Since $E^{(t)} \le A_f^{(t)} + B_f^{(t)} \le 1$, we have that:
\begin{equation*}
  B_f^{(t + 1)} = \norm{\Pi_f^\perp \psi_f^{(t + 1)}}^2 - A_f^{(t + 1)} = E^{(t)} \tproj_f^{(t)}+
  \norm{c_f^{\perp}}^2  \le
  B_f^{(t)}
  +
  \tproj_f^{(t)}
  -
  z_f^{(t)}
  \, .
\end{equation*}
where recall that the sequence $(z_f^{(t)})_t$ is equal to $z_f^{(t)} := B_f^{(t)} - \norm{c_f^{\perp}}^2 \in [0,B_f^{(t)}]$.

\paragraph{Quantum query case.} 
If the $(t+1)$-th query is quantum, according to the evolution
of the state definition (\cref{eq:def_state_evolution}), the algorithm states
after the $t + 1$ query are:
\begin{equation*}
    \begin{split}
    \psi_{f}^{(t + 1)} = U^{(t + 1)} Q_f \psi_{f}^{(t)} \ \ ; \ \ \phi^{(t + 1)} = U^{(t + 1)}  \phi^{(t)}.
    \end{split}
\end{equation*}
where $U^{(t + 1)}$ is a unitary independent of input, and $Q_f = I - 2 \Pi_f$.

Then, we have that:
\begin{equation*} \label{eq:a_f_t_plus_1_q}
    \begin{split}
    \sqrt{A_f^{(t+1)}} &= |\langle \phi^{(t+1)}, \psi_{f}^{(t+1)} \rangle| \\
    &= | \langle \phi^{(t)} U^{(t + 1)}, U^{(t + 1)} Q_f \psi_{f}^{(t)} \rangle | \\
    &= | \langle \phi^{(t)} ,  Q_f \psi_{f}^{(t)} \rangle | \\
    &= | \sqrt{A_f^{(t)}} - 2 \langle \Pi_f \phi^{(t)} , \Pi_f \psi_{f}^{(t)} \rangle |
    \end{split}
\end{equation*}

By the decomposition of $\psi_{f}^{(t)}$ (\cref{lemma:psidecompose}), we
know that:
\begin{equation*}
   \psi_{f}^{(t)} =  \sqrt{A_f^{(t)}} \cdot \phi^{(t)} + \omega \sqrt{B_f^{(t)} - || c ||^2 }  \cdot \phi^{(t)\perp} + c
\end{equation*}

where $|\omega| = 1$ and 
$c \perp \Phi = \text{span}\{\phi^{(t)}, \phi^{(t)\perp}\}$.
Then, we have:
\begin{equation*}
    \begin{split}
    \langle \Pi_f \phi^{(t)} , \Pi_f \psi_{f}^{(t)} \rangle 
    &= \langle \Pi_f \phi^{(t)}, \sqrt{A_f^{(t)}} \cdot \Pi_f \phi^{(t)} + \omega \sqrt{B_f^{(t)} - ||c||^2 }  \cdot \Pi_f \phi^{(t)\perp} + \Pi_f c \rangle \\
    &= \sqrt{A_f^{(t)}} \cdot \tproj_f^{(t)} + \omega \sqrt{B_f^{(t)} - || c ||^2} \langle \Pi_f \phi^{(t)}, \Pi_f \phi^{(t)\perp} \rangle
    \end{split}
\end{equation*}
where in the last equality we used that $\langle \phi^{(t)} | \Pi_f c \rangle = 0$,
Using the definition of the state $\phi^{(t)\perp}$ (Def.~\ref{def:2dplane}), we have that:
\begin{equation*}
    \begin{split}
   \Pi_f  \phi^{(t)\perp} &= \Pi_f \left( \phi_f^{(t)} \sqrt{1-\tproj_f^{(t)}} - \phi_f^{(t)\perp} \sqrt{\tproj_f^{(t)}} \right) = \sqrt{1-\tproj_f^{(t)}} \phi_f^{(t)} \\
       \langle \Pi_f \phi^{(t)}, \Pi_f \phi^{(t)\perp}  \rangle &= 
  \langle \sqrt{\tproj_f^{(t)}} \phi_f^{(t)} , \sqrt{1-\tproj_f^{(t)}} \phi_f^{(t)} \rangle =  \sqrt{\tproj_f^{(t)}} \cdot \sqrt{1-\tproj_f^{(t)}}
      \end{split}
\end{equation*}

As a result, we can rewrite $A_f^{(t + 1)}$ as:
\begin{equation*}
    \begin{split}
    \sqrt{A_f^{(t+1)}} 
    &= \left| (1 - 2 \tproj_f^{(t)}) \sqrt{A_f^{(t)}} - 2 \omega  \sqrt{B_f^{(t)} - || c ||^2} \cdot \sqrt{\tproj_f^{(t)}} \cdot \sqrt{1-\tproj_f^{(t)}} || \right| 
    \end{split}
\end{equation*}
Using the inequality $|a - b| \geq ||a| - |b|| \geq |a| - |b|$ for any $a, b$ complex numbers:
\begin{equation*}
    \begin{split}
        \sqrt{A_f^{(t+1)}} &\geq \left|1 -  2 \tproj_f^{(t)} \right| \sqrt{A_f^{(t)}} - 2 |\omega| \cdot \sqrt{\tproj_f^{(t)}} \cdot \sqrt{1-\tproj_f^{(t)}} \sqrt{B_f^{(t)} - || c ||^2} \\
        &\geq \left|1 -  2 \tproj_f^{(t)} \right| \sqrt{A_f^{(t)}} - 2 \cdot \sqrt{\tproj_f^{(t)}} \sqrt{B_f^{(t)}}
    \end{split}
\end{equation*}

Using that $\sqrt{A_f^{(t)}} \leq 1$ and the observation that if $x \geq a - b$, this implies that $x^2 \geq a (a - 2b)$ for any  $a, b > 0$, we can determine the 
lower bound on 
$A_f^{(t + 1)}$:
\begin{equation*}
    \begin{split}
    A_f^{(t+1)} &\geq \left|1 -  2\tproj_f^{(t)} \right| \sqrt{A_f^{(t)}} \cdot \left( \left|1 -  2 \tproj_f^{(t)} \right| \sqrt{A_f^{(t)}} - 4 \cdot \sqrt{\tproj_f^{(t)}} \sqrt{B_f^{(t)}} \right) \\
    &\geq \left(1 -  2 \tproj_f^{(t)} \right)^2 A_f^{(t)} - 4 \left|1 -  2 \tproj_f^{(t)} \right| \sqrt{\tproj_f^{(t)}} \sqrt{B_f^{(t)}} \\
    &\geq A_f^{(t)} -  4 \tproj_f^{(t)}  - 4 \sqrt{\tproj_f^{(t)}} \sqrt{B_f^{(t)}}
    \end{split}
\end{equation*}

As for a quantum query we have: $A_f^{(t+1)} + B_f^{(t+1)} = A_f^{(t)} + B_f^{(t)}$, we get:
\begin{equation*}
    \begin{split}
     B_f^{(t+1)} &= A_f^{(t)} + B_f^{(t)} - A_f^{(t+1)} \leq A_f^{(t)} + B_f^{(t)} -  A_f^{(t)} +  4 \tproj_f^{(t)} + 4 \sqrt{\tproj_f^{(t)}} \sqrt{B_f^{(t)}} \\
     &\leq  B_f^{(t)} +  4 \tproj_f^{(t)} + 4 \sqrt{\tproj_f^{(t)}} \sqrt{B_f^{(t)}}
    \end{split}
\end{equation*}

\end{proof}

\begin{lemma}[Progress Measures for
  $\dsearch$] \label{lemma:progress_general} For every $t$, the
  progress measures after the $t + 1$-th query satisfy the following
  recurrent relations:

  \begin{itemize}
  \item If the $t + 1$-th query is pseudo-classical, there exists
    $z_t \in [0, B^{(t)}]$ such that:
    \begin{equation}
      \begin{split}
        A^{(t + 1)} &\geq  A^{(t)} - 2 \solnl - 2 \sqrt{\solnl} \cdot \sqrt{z_t} \\
        B^{(t + 1)} &\leq  B^{(t)} - z_t + \solnl
        \end{split}
    \end{equation}
  \item If the $t + 1$-th query is quantum, then we have:
    \begin{equation}
        \begin{split}
            A^{(t + 1)} &\geq A^{(t)} - 4 \cdot \solnl - 4 \cdot \sqrt{\solnl} \cdot \sqrt{B^{(t)}} \\
            B^{(t + 1)} &\leq B^{(t)} + 4 \cdot \solnl + 4 \cdot \sqrt{\solnl} \cdot \sqrt{B^{(t)}}
        \end{split}
      \end{equation}
    \end{itemize}
\end{lemma}

\begin{proof}
    Letting $z_t \defeq \expt_{f\gets \distr}(z_f^{t})$, we can observe that
$z_t\in[0, B^{(t)}]$. Taking expectations over $\distr$, and
applying~\cref{lemma:expproduct}
($\expt \sqrt{g(Z)\cdot h(Z)} \le \sqrt{\expt g(Z) \cdot \expt h(Z)}$)
and~\cref{lemma:gammavssolnl} ($\gamma^{(t)} \le \solnl$), 
the relations for $A^{(t)}$ and $B^{(t)}$ follow.

\end{proof}

Next, since we intend to lower bound $A^{(\aquery)}$ and upper bound
$B^{(\aquery)}$, we can change the inequalities to equalities and
analyze instead the new sequences $(a_t,b_t)$ defined below. It is
clear that $A^{(\aquery)}\ge a_\aquery$ and
$B^{(\aquery)}\le b_\aquery$.

\begin{definition}[Sequences ${(a_t)}_{t \geq 0}, {(b_t)}_{t \geq
    0}$] \label{def:seq_a_b_general} We define the following sequences
  based on the evolution of the progress measures $A$ and $B$:
  \begin{align*}
    a_0 & \defeq A^{(0)} = 1\\
    b_0 & \defeq B^{(0)} = 0\\
    a_{t + 1} & \defeq 
                \begin{cases}
                  a_{t}  - 2 \cdot \solnl - 2 \cdot \sqrt{\solnl} \cdot
                  \sqrt{z_t} \, , \qquad  \text{if $t + 1 \in T_c$} \\
                  a_{t}  - 4 \cdot \solnl - 4 \cdot \sqrt{\solnl} \cdot \sqrt{b_{t}} \, , \qquad \text{if $t + 1 \in T_q$} 
                \end{cases}  \\
    b_{t + 1} & \defeq 
         \begin{cases}
           b_{t} + \solnl - z_t \, , \qquad \qquad \qquad \qquad  \text{if $t + 1 \in T_c$} \\
           b_{t} + 4 \cdot \solnl + 4 \cdot \sqrt{\solnl} \cdot
           \sqrt{b_{t}}  \, ,\qquad \text{if $t + 1 \in T_q$} 
         \end{cases}
  \end{align*}
  where $(z_t)_{t \geq
    1}$ is the sequence defined in the proof of
  Lemma~\ref{lemma:progress_general}, which satisfies
$0 \leq z_t \leq B^{(t)}$ for any $t$.
\end{definition}

\begin{lemma}[Bounding $a_\aquery$ and
  $b_{\aquery}$]
  \label{lemma:bound_ab}
  \begin{equation}
    a_\aquery \ge 1 - 4 \solnl \cdot (\sqrt{\cquery} + \qquery)^2\, , \qquad b_{\aquery} \le \solnl \cdot (\sqrt{\cquery} + 2
    \qquery)^2 \, .
  \end{equation}
\end{lemma}

\begin{proof}

The proof consists of four steps.

\medskip
{\bf (1)} First we show that $b_\aquery \le \left( \sqrt{\tau_c} + 2 \tau_q \right)^2 \cdot \solnl$.

To get an upper bound for each term of this sequence, we can let $z_t = 0$ and instead consider the sequence:
\begin{equation*}
    \begin{split}
    d_{t + 1} \defeq  
         \begin{cases}
            d_{t} + \solnl \, ,\qquad \qquad \qquad \qquad \qquad \text{ if $t + 1 \in T_c$} \\
            d_{t} + 4 \cdot  \solnl + 4 \cdot \sqrt{\solnl} \cdot  \sqrt{d_{t}} \, ,\qquad \text{ if $t + 1 \in T_q$} 
        \end{cases}
    \end{split}
\end{equation*}
As a result we have: $b_t \leq d_t$ for any $t \in [\tau]$.

Our task is to bound the last term $d_{\tau}$ in the sequence. Every hybrid strategy $A$ that uses $\tau_c$ classical queries and $\tau_q$ quantum queries can be expressed by $A = [x_1, \cdots , x_{\tau}]$, where if $x_i = 0$ (resp. $x_i=1$) indicates that the $i$-th query of $A$ is classical (resp. quantum), and there are exactly $\tau_c$ values of $0$ and $\tau_q$ values of $1$. Therefore, the sequence $(d_t)_t$ parameterized by the strategy $A$, denoted as $(d_t^A)_t$, can be re-written as:
\begin{equation}\label{eq:d_t_a}
    \begin{split}
    d_{t + 1}^A \defeq  
         \begin{cases}
            d_{t}^A + \solnl \, ,\qquad \qquad \qquad \qquad \qquad \text{ if $x_{t + 1} = 0$} \\
            d_{t}^A + 4 \cdot  \solnl + 4 \cdot \sqrt{\solnl} \cdot  \sqrt{d_{t}} \, ,\qquad \text{ if $x_{t + 1} = 1$} 
        \end{cases}
    \end{split}
\end{equation}
 
Our task then becomes determining the strategy $A^*$ which achieves the maximum $d_{\tau}^{A^*}$. We claim that 
\begin{equation*}
    A^* \defeq  [0, \cdots , 0, 1, \cdots , 1] \, ,
\end{equation*} 
namely the strategy of making all classical queries upfront is optimal. This follows from a greedy argument. 

Consider two arbitrary strategies $A = [x_1, \cdots, x_i, x_{i + 1}, \cdots, x_{\tau}]$ and 
$B = [y_1, \cdots , y_i,$ $ y_{i + 1}, \cdots, y_{\tau}]$
which only differ in the $i$ and $i + 1$-th queries. Namely, $x_i = 0$, $x_{i + 1} = 1$ and $y_i = 1$, $y_{i + 1} = 0$ and $x_j = y_j$ for $j \in \{1, \cdots , \tau\} - \{i, i + 1\}$. We next show that $d_{\tau}^A > d_{\tau}^B$. As $x_1 = y_1, \cdots x_{i - 1} = y_{i - 1}$, this implies directly that $d_{i - 1}^A = d_{i - 1}^B$. Then for the $i$-th and $i + 1$ terms of the two sequences we have:
\begin{equation*}
    \begin{split}
        d_i^A &= d_{i - 1}^A + \solnl \qquad \qquad \qquad \ \ \ \ \ \ ; \ \ \ \ d_{i+1}^A = d_{i - 1}^A + 5 \solnl + 4 \sqrt{\solnl} \sqrt{d_{i - 1}^A + \solnl} \\
        d_i^B &= d_{i - 1}^B + 4 \solnl + 4 \sqrt{\solnl} \sqrt{d_{i - 1}^B} \ \ \ ; \ \ \ d_{i+1}^B = d_{i - 1}^B + 5 \solnl + 4 \sqrt{\solnl} \sqrt{d_{i - 1}^B}
    \end{split}
\end{equation*}
Then, as $d_{i - 1}^A = d_{i - 1}^B$ it is clear that $d_{i+1}^A > d_{i+1}^B$. As $x_j = y_j$ for all $i + 2 \leq j \leq \tau$, this also implies that $d_{\tau}^A > d_{\tau}^B$. 

Denote the following swap operation on strategies. Given as input a strategy $A = [x_1, ..., x_i, x_{i + 1}, \cdots, x_{\tau}]$ the function $\mathsf{swap}_i$ outputs a strategy $A'$:
\begin{equation*}
    \mathsf{swap}_i(A) = A' \text{ where } A' = [x_1, ..., x_{i + 1}, x_i, \cdots, x_{\tau}]
\end{equation*}
Our previous argument implies that for a strategy $A$ such that $x_i = 0$ and $x_{i + 1} = 1$, we have: $d_{\tau}^{A} > d_{\tau}^{\mathsf{swap}_i(A)}$. Therefore, we can see that any strategy $A = [x_1, ..., x_{\tau}]$ can be obtained from a sequence of applications of $\mathsf{swap}_i$ on $A^*$.
\begin{equation*}
    A^* \defeq  [0, \cdots , 0, 1, \cdots , 1] \xrightarrow{\mathsf{swap}_{i_1}} \cdots \xrightarrow{\mathsf{swap}_{i_k}} A \ \ \text{for some indices $i_1, ..., i_k$} \, .
\end{equation*}
It hence follows that $d_{\tau}^{A^*} \geq d_{\tau}^{A}$, i.e., $A^*$ is the optimal strategy. 

Now, let us compute the last term of the optimal strategy, i.e.: $d_{\tau}^{A^*}$. We can rewrite the sequence $d_t$ as:
\begin{equation*}
    \begin{split}
    d_{t + 1}^{A^*} = 
         \begin{cases}
            d_{t}^{A^*} + \solnl  \, ,\qquad \qquad \qquad \qquad \qquad \qquad \qquad \qquad \qquad \qquad  \text{ if $0 \leq t < \tau_c$} \\
            d_{t}^{A^*} + 4 \cdot  \solnl + 4 \cdot \sqrt{\solnl} \cdot  \sqrt{d_{t}^{A^*}} = \left(\sqrt{d_{t}^{A^*}} + 2 \sqrt{\solnl} \right)^2  \, , \ \ \ \text{ if $\tau_c \leq t < \tau$} 
        \end{cases}
    \end{split}
\end{equation*}
As $d_0^{A^*} = 0$, it is clear that we have: $d_{\tau_c}^{A^*} = \tau_c \cdot \solnl$. 
For $\tau_c \leq t \leq \tau$, we will prove by induction that:
\begin{equation*}
    d_{t}^{A^*} = \left( \sqrt{\tau_c} + 2 (t - \tau_c) \right)^2 \cdot \solnl 
\end{equation*}
For the base case $t = \tau_c$, we already showed that $d_{\tau_c}^{A^*} = \tau_c \cdot \solnl$. 
For the inductive step, we have that:
\begin{equation*}
    \begin{split}
    d_{t + 1}^{A^*} 
                    &= \left(  \sqrt{\left( \sqrt{\tau_c} + 2 (t - \tau_c) \right)^2 \cdot \solnl}  + 2 \sqrt{\solnl} \right)^2  \\
                    &= \left( \sqrt{\tau_c} + 2(t - \tau_c + 1) \right) \cdot \solnl 
    \end{split}
\end{equation*}
which concludes the inductive proof. Hence, by putting things together we have:
\begin{equation}\label{eq:b_tau}
    b_{\tau} \leq d_{\tau} \leq d_{\tau}^{A^*} = \left( \sqrt{\tau_c} + 2 \tau_q \right)^2 \cdot \solnl
\end{equation}

\medskip
{\bf (2)} Secondly, we show that  $\sum_{t \in T_q} \sqrt{b_{t - 1}} \leq \sqrt{\solnl} \cdot  \tau_q(\sqrt{\tau_c} + \tau_q - 1)$.

As for $b_{\tau}$, to get an upper bound we let $z_t = 0$ and use the sequence $(d_t^A)_t$. From the definition of the sequence (Equation~\ref{eq:d_t_a}), 
 it is clear that $(d_t^A)_t$ is a strictly increasing sequence for any strategy $A$. This also implies that for any strategy $A$ we have:
 \begin{equation*}
     \sum_{t \in T_q} \sqrt{d_{t - 1}^A} \leq \sum_{\tau_c \leq t \leq \tau} \sqrt{d_t^A}
 \end{equation*}

  In other words, $\sum_{t \in T_q} \sqrt{d_{t - 1}^A}$ is maximized when the strategy performs first all $\tau_c$ classical queries and then the $\tau_q$ quantum queries. Hence, the maximum is achieved for the strategy described above by the sequence $(d_t^{A^*})_t$.

Using the previous result in Equation~\ref{eq:b_tau}:

\begin{equation*}
     \sum_{\tau_c \leq t \leq \tau} d_{t}^{A^*} = \solnl \cdot \sum_{\tau_c \leq t \leq \tau} \left( \sqrt{\tau_c} + 2 (t - \tau_c) \right)^2   
\end{equation*}

This gives us:
\begin{align*}
  \sum_{t \in T_q} \sqrt{b_{t - 1}} \leq 
  \sum_{\cquery \leq t \leq \aquery} \sqrt{d_{t}^{A^*}} & =
                                                          \sqrt{\solnl} \sum_{\cquery \leq t \leq \aquery}  \sqrt{\cquery} +
                                                          2 (t - \cquery) \\
                                                        & \leq \sqrt{\solnl} \left( \qquery (\sqrt{\cquery} - 2 \cquery) + 2 \sum_{\cquery \leq t \leq \aquery} t \right) \\
                                                        & =
                                                          \sqrt{\solnl}
                                                          \left(
                                                          \qquery
                                                          \sqrt{\cquery}
                                                          (1 -
                                                          2\sqrt{\cquery}
                                                          +
                                                          2\sqrt{\cquery})
                                                          + 2 \cdot
                                                          \frac{(\qquery
                                                          -
                                                          1)\qquery}{2}
                                                          \right) \\
  & =  \sqrt{\solnl} \qquery (\sqrt{\cquery} + \qquery - 1)
\end{align*}

{\bf (3)} Thirdly, we show that  $\sum_{t \in T_c} \sqrt{z_{t-1}} \leq \sqrt{\solnl} \cdot (\tau_c + 2\sqrt{\tau_c} \tau_q)$.

By definition of the sequence $z_t$ (Def.~\ref{def:seq_a_b_general}), we know that for $t \in T_c$:

\begin{equation*}
    \sum_{t \in T_c} z_{t - 1} = \solnl \cdot \tau_c + \sum_{t \in T_c} (b_{t - 1} - b_{t})
\end{equation*}

It hence suffices to derive an upper bound on $\sum_{t \in T_c} (b_{t - 1} - b_{t})$. We can rewrite $b_{\tau}$ as:
\begin{equation*}
    \begin{split}
    b_{\tau} &= 
     b_0 + \sum_{t = 1}^{\tau} (b_t - b_{t - 1}) 
        = \sum_{b_t \geq b_{t - 1}} (b_t - b_{t - 1}) + \sum_{b_t < b_{t - 1}} (b_t - b_{t - 1})
    \end{split}
\end{equation*}

As a result, we  have that:
\begin{equation*}
    \begin{split}
        \sum_{t \in T_c \ \wedge \ b_t < b_{t - 1}} (b_{t - 1} - b_{t}) < \sum_{b_t < b_{t - 1}} (b_{t - 1} - b_{t}) 
        = \sum_{b_t \geq b_{t - 1}} (b_t - b_{t - 1}) - b_{\tau}
    \end{split}
\end{equation*}
In other words we also have:
\begin{equation*}
    \sum_{t \in T_c  \ \wedge \ b_t < b_{t - 1}} (b_{t - 1} - b_{t}) < 
    \sum_{t \in T_c \ \wedge \ b_t \geq b_{t - 1}} (b_t - b_{t - 1}) + \sum_{t \in T_q \ \wedge \  b_t \geq b_{t - 1}} (b_t - b_{t - 1})
\end{equation*}

For $t \in T_q$, from the sequence definition (Def.~\ref{def:seq_a_b_general}), we have that 
$b_t > b_{t - 1}$ and hence:
\begin{equation*}
    \begin{split}
    \sum_{t \in T_c  \ \wedge \  b_t < b_{t - 1}} (b_{t - 1} - b_{t}) 
    < \sum_{t \in T_c  \ \wedge \  b_t \geq b_{t - 1}} (b_t - b_{t - 1}) + 4 \tau_q \cdot \solnl + 4 \sqrt{\solnl} \sum_{t \in T_q} \sqrt{b_{t - 1}}
    \end{split}
\end{equation*}

By applying step (2), we get:
\begin{equation*}
     \sum_{t \in T_c  \ \wedge \  b_t < b_{t - 1}} (b_{t - 1} - b_{t}) <  
     \sum_{t \in T_c  \ \wedge \  b_t \geq b_{t - 1}} (b_t - b_{t - 1}) + 4 \solnl \tau_q  + 4 \solnl \tau_q(\sqrt{\tau_c} + \tau_q - 1)
\end{equation*}
By subtracting the first sum from the right hand side we get: 

\begin{equation*}
     \sum_{t \in T_c} z_{t - 1} =  \solnl \cdot \tau_c + \sum_{t \in T_c} (b_{t - 1} - b_{t}) 
     < \solnl \cdot \left(\tau_c + 4 \tau_q^2 + 4 \tau_q \sqrt{\tau_c}  \right)
\end{equation*}

Finally, by using the Cauchy-Schwarz inequality: 
\begin{equation*}
    \begin{split}
     \sum_{t \in T_c} \sqrt{z_{t - 1}} \leq \sqrt{ \solnl \cdot \left(\tau_c + 4 \tau_q^2 + 4 \tau_q \sqrt{\tau_c}  \right)} \cdot \sqrt{\tau_c} 
                                \leq \sqrt{\solnl} \cdot (\tau_c + 2\tau_q \sqrt{\tau_c})
    \end{split}
\end{equation*}

\medskip 
{\bf (4)} In the final step, we show that $a_\aquery \ge 1 - 4 \solnl (\sqrt{\tau_c} + \tau_q)^2$ .  

From the definition of $a_{t}$ (Def.~\ref{def:seq_a_b_general}):
\begin{align*}
  a_{\tau} 
           & = a_0 + \sum_{t = 1}^{\tau} (a_t - a_{t - 1}) \\
           & = 1 - \sum_{t \in T_c} \left( 2\solnl + 2 \sqrt{\solnl} \cdot \sqrt{z_{t - 1}} \right) - \sum_{t \in T_q} \left( 4 \solnl + 4 \sqrt{\solnl} \cdot \sqrt{b_{t - 1}} \right) \\
           & = 1 - 2 \tau_c \solnl - 4 \tau_q \solnl - 2 \sqrt{\solnl} \sum_{t \in T_c} \sqrt{z_{t - 1}} - 4 \sqrt{\solnl} \sum_{t \in T_q} \sqrt{b_{t - 1}} \\
\end{align*}
Using the bounds derived in steps (2) and (3), we get :
\begin{align*}
  a_{\tau} &\geq 1 - 2 \tau_c \solnl - 4 \tau_q \solnl - 2 \solnl \cdot (\tau_c + 2\sqrt{\tau_c} \tau_q) - 4 \solnl \cdot \tau_q(\sqrt{\tau_c} + \tau_q - 1) \\
        &= 1 - 4 \solnl (\sqrt{\tau_c} + \tau_q)^2
\end{align*}
\end{proof}

\section{Applications of Bernoulli Search}
\label{sec:apps}
\subsection{Generic Security of Hash Functions against Hybrid Adversaries}
\label{sec:hash} 

In this section we apply our hardness result on the Bernoulli Search problem to establish important security properties of hash functions against hybrid adversaries. We adapt the techniques developed in \cite{HRS16} against full quantum adversaries, and our proof proceeds as follows:
\begin{tiret}
\item Reduce the hash security properties to the Bernoulli Search problem. More concretely, show how an instance of the Bernoulli Search problem can be turned into an instance of the security experiment.
\item Determine the number of classical and quantum queries in the
  reduction.
\item Apply our hardness bound on the Bernoulli Search problem  (Theorem~\ref{thm:bernoulli}) to bound the success probability of breaking the security properties of hash functions.
\end{tiret}

We now formally introduce the security properties of hash functions to be analyzed in the hybrid model.

\subsubsection{Hash Functions Background}

Let $n \in \mathbb{N}$ be the security parameter, $m=\poly(n)$, $k=\poly(n)$, and $\cH_n :=\{H_K : \bool^m \rightarrow \bool^n\}_{K\in \bool^k}$ be a family of hash functions, where $K$ denotes the index of the hash function. We will denote by $M$ the input of the hash function.

\begin{definition}[$\ow$] For any $\qpt$ adversary $\cA$, we define the probability of success of breaking the one-wayness ($\ow$) of a family of hash functions ${\mathcal{H}_n}$ as: 
\begin{equation*}
    \begin{split}
    \succp^{\ow}_{\mathcal{H}_n}(\cA) &= \Pr[ K \leftarrow \bool^k \ , \ M  \leftarrow \bool^m,Y \leftarrow H_K(M) \ ; \\
    & M' \leftarrow \cA(K,Y) : Y = H_K(M')]
    \end{split}
\end{equation*}
Similarly, we define single-function, multi-target preimage resistance ($\smow$):
\begin{equation*}
    \begin{split}
    \succp^{p-\smow}_{\mathcal{H}_n}(\cA) &= \Pr[K \leftarrow \bool^k \ , \ M_i \leftarrow \bool^m \ , Y_i \leftarrow H_K(M_i) \, \ 0 < i \leq p; \\
    & M' \leftarrow \cA(K, (Y_1, \cdots ,Y_p)) : \exists 0 < i \leq p, Y_i = H_K(M')]
    \end{split}
\end{equation*}
And we also define multi-function, multi-target preimage resistance ($\mmow$): 
\begin{equation*}
    \begin{split}
    \succp^{\mmow}_{\mathcal{H}_n}(\cA) &= \Pr[K_i \leftarrow \bool^k \ , \ M_i \leftarrow \bool^m \ , \ Y_i \leftarrow H_{K_i}(M_i) \ , \ 0 < i \leq p  \\
    & (j, M') \leftarrow \cA((K_1, Y_1), \cdots , (K_p, Y_p)) \ : \ Y_j = H_{K_j}(M')]
    \end{split}
\end{equation*}
\end{definition}

\begin{definition}[$\spr$] For any $\qpt$ adversary $\cA = (\cA_1, \cA_2)$, we define the probability of success of breaking the second-preimage resistance ($\spr$)  of a family of hash functions ${\mathcal{H}_n}$ as: 
    \begin{equation*}
    \begin{split}
    \succp^{\spr}_{\mathcal{H}_n}(\cA) &= \Pr[ K \leftarrow \bool^k \ , \ M  \leftarrow \bool^m \ ; \\
    & M' \leftarrow \cA(K,M) : M' \neq M \text{ and } H_K(M) = H_K(M')]
    \end{split}
\end{equation*}
\end{definition}

\begin{definition}[$\etcr$] For any $\qpt$ adversary $\cA$, we define the probability of success of breaking the extended target collision-resistance ($\etcr$) of a family of hash functions ${\mathcal{H}_n}$ as: 
        \begin{equation*}
    \begin{split}
    \succp^{\etcr}_{\mathcal{H}_n}(\cA) &= \Pr[ M \leftarrow \cA_1(1^n) \ , \ K  \leftarrow \bool^k  ; \\
    & (M', K') \leftarrow \cA_2(K, M) : M' \neq M \text{ and } H_K(M) = H_{K'}(M')]
    \end{split}
\end{equation*}
\end{definition}
Second-preimage resistance and extended target collision-resistance can be extended to the $\mathsf{SM}$ (single-function, multi-target) and $\mathsf{MM}$ (multi-function, multi-target) versions analogous to the one-wayness case.


\subsubsection{Hybrid Security of Hash Functions}

\begin{lemma}[Hybrid Security of $\ow$] \label{lemma:proof_ow}
Let $m = cn$ for $c > 1$ constant and $p = o(2^n)$. For any \hybrid{} algorithm $\cA$ with $\tau_c$ classical queries and $\tau_q$ quantum queries we have:
\begin{align*}
  \succp^{\ow}_{\mathcal{H}_n}(\cA) & \leq \frac{1}{2^n} \cdot (2\sqrt{\cquery} + 4
                                      \qquery + 1)^2 \, ,\\
  \succp^{\smow}_{\mathcal{H}_n}(\cA) & \leq p \cdot \frac{1}{2^n} \cdot (2\sqrt{\cquery} + 4
                                        \qquery + 1)^2  \, , \\
  \succp^{\mmow}_{\mathcal{H}_n}(\cA) & \leq \frac{1}{2^n} \cdot (2\sqrt{\cquery} + 4
                                        \qquery + 1)^2 \, .
\end{align*}
\end{lemma}

\begin{lemma}[Hybrid Security of $\spr$] For any \hybrid{} algorithm
  $\cA$ with $\tau_c$ classical queries and $\tau_q$ quantum queries
  we have:
  \begin{align*}
    \succp^{\spr}_{\mathcal{H}_n}(\cA) & \leq \frac{1}{2^n} \cdot (2\sqrt{\cquery} + 4
                                         \qquery + 1)^2 \, , \\
    \succp^{\smspr}_{\mathcal{H}_n}(\cA) & \leq p \cdot \frac{1}{2^n} \cdot (2\sqrt{\cquery} + 4
                                           \qquery + 1)^2 \, , \\
    \succp^{\mmspr}_{\mathcal{H}_n}(\cA) & \leq \frac{1}{2^n} \cdot (2\sqrt{\cquery} + 4
                                           \qquery + 1)^2 \, .
  \end{align*}
\end{lemma}

\begin{lemma}[Hybrid Security of $\etcr$] For any \hybrid{} algorithm $\cA$ with $\tau_c$ classical queries and $\tau_q$ quantum queries we have:
\begin{align*}
  \succp^{\etcr}_{\mathcal{H}_n}(\cA) & \leq \frac{1}{2^n} \cdot (2\sqrt{\cquery} + 4
                                        \qquery + 1)^2 + \frac{8(\tau_c + \tau_q)^2}{2^k} \, , \\
  \succp^{\mmetcr}_{\mathcal{H}_n}(\cA) & \leq p \cdot \left(\frac{1}{2^n} \cdot (2\sqrt{\cquery} + 4 
                                          \qquery + 1)^2 +
                                          \frac{8(\tau_c +
                                          \tau_q)^2}{2^k} \right) \, .
\end{align*}
\end{lemma}

The results are compared with the classical and quantum adversary settings in~\cref{tab:comparison}.

\begin{table}[H]
\begin{center}
\setlength\tabcolsep{4.0pt}
 \begin{tabular}{||c | c | c | c ||} 
 \multicolumn{1}{c}{}  &  \multicolumn{1}{c}{$\begin{array} {lcl} \text{Classical} \\ \text{Adversary} \end{array}$} &   \multicolumn{1}{c}{$\begin{array} {lcl} \text{Quantum} \\ \text{Adversary} \end{array}$} &  \multicolumn{1}{c}{$\begin{array} {lcl} \text{Hybrid} \\ \text{Adversary} \end{array}$}  \\ 
 \hline\hline
$\begin{array} {lcl} \ow, \mmow, \\ \spr, \mmspr \end{array}$ & $\frac{q + 1}{2^n}$ & $O(\frac{(q + 1)^2}{2^n})$ & $O(\frac{(\sqrt{\tau_c} + \tau_q + 1)^2}{2^n})$ \\ 
 \hline
$\begin{array} {lcl} \smow, \\ \smspr \end{array}$ & $p \cdot \frac{q + 1}{2^n}$ & $O(p\frac{(q + 1)^2}{2^n})$ & $O(\frac{p(\sqrt{\tau_c} + \tau_q + 1)^2}{2^n})$ \\
 \hline
$\etcr$ & $\frac{q + 1}{2^n} + \frac{q}{2^k}$ & $O(\frac{(q + 1)^2}{2^n} + \frac{q^2}{2^k})$ & $O(\frac{(\sqrt{\tau_c} + \tau_q + 1)^2}{2^n} + \frac{(\tau_c + \tau_q)^2}{2^k})$ \\
\hline
$\begin{array} {lcl} \mmetcr \end{array}$ & $p \cdot \left( \frac{q + 1}{2^n} + \frac{q}{2^k} \right)$ & $O(p\frac{(q + 1)^2}{2^n} + p\frac{q^2}{2^k})$ & $O(\frac{p(\sqrt{\tau_c} + \tau_q + 1)^2}{2^n} + \frac{p(\tau_c + \tau_q)^2}{2^k})$ \\ 
\hline 
\end{tabular}
\end{center}
 \caption{Security of hash functions $\cH_n :=\{H_K : \bool^m \rightarrow \bool^n\}_{K\in \bool^k}$ against generic classical, quantum and hybrid adversaries. Entries represent the probability of success of classical adversaries equipped with $q$ classical queries, quantum adversaries equipped with $q$ quantum queries, respectively \hybrid{} adversaries equipped with $\tau_c$ classical and $\tau_q$ quantum queries.}
\label{tab:comparison}
\end{table}

\vspace{-2em}

We give an example proof for the one-wayness, the rest follows similarly. 

\begin{proof}[Proof of~\cref{lemma:proof_ow} ($\ow$)]
  Given an Bernoulli Search instance, we will show how to construct an
  instance of one-wayness ($\ow$).

\begin{pcvstack}[boxed,center,space=1em]
\procedureblock[linenumbering]
{Bernoulli Search to $\ow$ Reduction }{
 	\textbf{Input}: f : \bool^m \rightarrow \bool \text{ sampled from distribution } \distr_\berpar. \text{ Set } \berpar = \frac{1}{2^n}; \\
        \text{Sample uniformly } y \in \bool^n; \\
        \text{Let random function } g : \bool^m \rightarrow \bool^n - \{y\}; \\
        \text{Construct function } G: \bool^m \rightarrow \bool^n \text{ defined as:} \\
          G(x) =  y \ \ \ \ \ \ \ \text{, if } f(x) = 1 \\
          G(x) = g(x) \ \ \ \text{, else} \\
        \textbf{Output}: \text{$\ow$ instance } (y, G).
  }
  
  \procedure[linenumbering,mode=text]{$\ow$ Adversary}{
 \textbf{Input}: Given $y$ and oracle access to $G$ \\
 \textbf{Task}: Find $x \in \bool^m$ such that $G(x) = y$
 }
\end{pcvstack}

\emph{Analysis of the reduction.} We first argue that the reduction outputs $(y, G)$ a valid instance of the $\ow$ experiment except with negligible discrepancy. Observe that $(y, G)$ is distributed identically to the distribution $D_1 = \{(z, H)\}$, where $z$ is sampled uniformly at random from $\bool^n$ and $H : \bool^m \rightarrow \bool^n$ random function. On the other hand, the distribution $D_0 := \{(H(x), H)\}$ in the $\ow$ experiment is obtained by choosing $H$ uniformly at random and $x$ uniformly from domain $\bool^m$. Then:
\begin{equation*}
  \mathsf{SD}(D_0, D_1) = \frac{1}{2} \sum_{z, H} \left| \Pr_{z, H}[z, H] - \Pr_{x, H}[H(x), H] \right| 
            = \frac{1}{2} \sum_{z} \sum_{H} \frac{1}{|\cH|} \left|\frac{|H^{-1}(z)}{2^m} - \frac{1}{2^n} \right| \, .
    \end{equation*}
    By Jensen's inequality we get that $\mathsf{SD}(D_0, D_1) \leq \frac{1}{2} \cdot \sqrt{\frac{2^n}{2^m}}$, which is
   negligibly small when setting $m \geq 2n$. 

\emph{Implementation of $G$ using $f$.} Now we show how to implement oracle $G$ using oracle access to $f$ and knowledge of $g$ and $y$. 
\begin{itemize}
\item Quantum query to $G$:
  $\sum_{x, z} a_{x, z} \ket{x, z} \xrightarrow{G} \sum_{x, z} a_{x,
    z} \ket{x, z + G(x)}$.
  \begin{align*}
    & \sum_{x, z} a_{x, z} \ket{x} \ket{z} \ket{0} \quad \text{(initial state to be queried) } \\
    \mapsto & \sum_{x, z} a_{x, z} \ket{x} \ket{z} \ket{f(x)} \quad \text{(evaluate $f$)}\\
    \mapsto & \sum_{x, z} a_{x, z} \ket{x} \ket{z + f(x) \cdot y + (1 - f(x)) \cdot g(x)} \ket{f(x)} \\
    = & \sum_{x, z} a_{x, z} \ket{x} \ket{z + G(x)} \ket{f(x)}  \\
    \mapsto & \sum_{x, z} a_{x, z} \ket{x} \ket{z + G(x)} \ket{0} \quad \text{(uncompute $f$)}
  \end{align*}
  It is clear that a quantum query to $G$ requires two quantum queries
  to the Bernoulli function $f$.
\item Classical query to $G$:
  \begin{equation*}
    x \mapsto f(x) \cdot y + (1 - f(x)) \cdot g(x) = G(x)
    \, .
  \end{equation*}
A classical query to $G$ requires a classical query to Bernoulli function $f$.
\end{itemize}
Therefore, a hybrid adversary against $\ow$ with $\tau_c$ classical queries and $\tau_q$ quantum queries would give rise to a hybrid algorithm that solves the Bernoulli Search problem using $\tau_c$ classical queries and $2 \cdot \tau_q$ quantum queries. By our hardness bound on the Bernoulli Search (Theorem~\ref{thm:bernoulli}, with $\berpar = \frac{1}{2^n}$), we conclude that:
\begin{equation*}
    \succp^{\ow}_{\mathcal{H}_n}(\cA) \leq \frac{1}{2^n} \cdot (2\sqrt{\cquery} + 4
    \qquery + 1)^2 \, .
\end{equation*}
\end{proof}

\subsection{The Bitcoin Blockchain 
in the Presence of Hybrid Adversaries}
\label{sec:pqbc_hybrid} 
In this section, we apply our hardness bound on the Bernoulli search to derive a query complexity of 
proofs of work in the Bitcoin blockchain against hybrid adversaries, which in turn enables us to establish the \emph{hybrid} security of the Bitcoin
backbone protocol based on existing analysis in the classical and fully quantum setting~\cite{GKL15,CGKSW23}.

\subsubsection{The Bitcoin Backbone Protocol}

A {\em proof of work} (\pow)
enables a party to convince other parties that considerable effort has
been invested in solving a computational task.  In the blockchain setting, the objective of a \pow{} is to confirm new
transactions to be included in the blockchain.
To successfully create a \pow{} in Bitcoin, one needs to find a value
(``witness'') such that evaluating a hash function (SHA-256) on this
value together with (the hash of) the last block and new transactions
to be incorporated, yields an output below a threshold. A party who
produces such a \pow{} gets to append a new block to the blockchain
and is rewarded. A {\em blockchain} hence consists of a sequence of
such \emph{blocks}. Each party maintains such a blockchain, and
attempts to extend it via solving a \pow.

\begin{definition}[Blockchain \pow---Informal]
\label{def:blockchain-pow}
Given a hash function $h$, a positive integer $T$, and a string $z$
  representing the hash value of the previous block, the goal is to
  find a value $ctr$ such that $h(ctr, z) \le T$.
\end{definition}

It is shown
in~\cite{GKL15} that the blockchain data structure built by the
Bitcoin backbone protocol satisfies a number of basic properties when $h$ is modelled as a random oracle. One important one, called {\em common prefix}, requires the existence and persistence of a common prefix
of blocks among the chains of honest parties. Another one, called
\emph{chain quality}, stipulates the proportion of honest blocks in
any portion of some honest party’s chain.

\begin{definition}[Common Prefix]
The {\em common prefix} property with parameter $k \in \mathbb{N}$, 
states that for any pair of honest players $P_1, P_2$ adopting chains $\chain_1, \chain_2$ at rounds $r_1 \leq r_2$, it holds that $\chain_1^{\lceil k} \preceq \chain_2$ (the chain resulting from pruning the $k$ rightmost blocks of $\chain_1$ is a prefix of $\chain_2$).
\end{definition}

\begin{definition}[Chain Quality] \label{def:chain_quality}
The {\em chain quality} property with parameters $\mu \in \mathbb{R}$ and $l \in \mathbb{N}$, states that for any honest party $P$ with
chain $\chain$, it holds that for any $l$ consecutive blocks of $\chain$, the ratio of blocks created by honest players is at least $\mu$.
\end{definition}

\paragraph{Parameters and random variables.}
Next, we recall some important notions in the Bitcoin backbone protocol setting.
\begin{itemize}[noitemsep]
\item $\tau_c$ and $\tau_q$ denotes the number of adversarial classical queries respectively quantum queries per round;
\item $f$ is the probability that at least one honest party generates a \pow{}
in a round;
\item $\epsilon$ will be used for the concentration quality of random variables;
\item $\kappa$ denotes the security parameter;
\item $k$ denotes the number of blocks for common prefix property and $\mu$ denotes the chain quality parameter;
\item $s$ refers to the total number of rounds;
\item $p = \frac{T}{2^{\kappa}}$, where $T$ denotes the difficulty parameter for solving a \pow{}. $p$ can be understood as the probability of success of generating a \pow{} using a single classical query;
\item $f$ denotes the
probability that at least one honest player generates a \pow{} in a single
round (e.g., in the Bitcoin system, $f$ is about $2-3\%$).

\end{itemize}

\subsubsection{Hybrid Security of the Bitcoin Backbone Protocol}

\begin{theorem}
Under the following condition, which we call the \emph{hybrid honest-majority} condition:
\begin{equation*}
    \sqrt{\tau_c} + 2 \tau_q \leq \frac{1}{\sqrt{f(1 - f)p}} \cdot \mathsf{negl}(\kappa) \, ,
\end{equation*}
the desired properties of a blockchain hold with probability $1 - \mathsf{negl}(\kappa)$:
\begin{tiret}
\item The common prefix property of the Bitcoin backbone protocol holds with parameter $k \geq 2sf$, for any $s \geq \frac{2}{f}$ consecutive rounds; 
\item the chain quality property holds with parameter $l \geq 2sf$ and ratio of honest blocks $\mu$ with $\mu = f$.
\end{tiret}
\end{theorem}

\begin{proof}

From \cite{CGKSW23} it is known that the common prefix and chain quality properties hold as long as in any round, the number of solved $\pow$s using $\tau_c$ classical queries and $\tau_q$ quantum queries is at most $E := (1 - \epsilon)f (1 - f)$.
It should be clear from 
Definition~\ref{def:blockchain-pow} 
that solving a single $\pow$ is equivalent to solving the Bernoulli Search problem with distribution $\distr_\berpar$, where we set $\eta = p = \frac{T}{2^{\kappa}}$.
Then, the probability that a \hybrid{} algorithm equipped with $\tau_c$ classical and $\tau_q$ quantum queries to solve $E$ $\pow$s is at most:
\begin{equation*}
    P_{\succp}^{\pow} \leq E \cdot  \succp_{\adv,\distr_\eta} \leq E \cdot p \cdot (2\sqrt{\cquery} + 4
    \qquery + 1)^2.
\end{equation*}
As a result, the probability that the two security properties hold under the hybrid honest-majority is:
\begin{equation*}
    \begin{split}
    P_{\mathsf{sec}} &= 1 - P_{\succp}^{\pow} \geq 1 - (1 - \epsilon)f (1 - f) p  \cdot (2\sqrt{\cquery} + 4  \qquery + 1)^2 \\
                    &\geq 1 - \mathsf{negl}(\kappa)
    \end{split}
\end{equation*}

\end{proof}

\section{A Quantum Algorithm for Distributional Search} %
\label{sec:quantum_algorithm}
In this section, we describe a quantum search algorithm that is takes
into account a given distribution $\distr$ by adapting Grover's
algorithm. For convenience, we first introduce some new notations, and
derive an alternative characterization the $\solnl$ parameter. \\

\subsection{Characterization of \texorpdfstring{\solnl}{}} %
\label{sec:nu_D}

Recall that
$\solnl = \max_{\phi} \expt_{f \leftarrow \distr}{\norm{\pi_f
    \phi}^2}$. We write down the truth table to represent each
$f: [m] \rightarrow \bool$ as a bit-string $x \in \bool^m$. $\distr$
becomes a distribution on $\bool^m$ and we write $d_x := D(x)$. Then
\begin{align*}
  \solnl = \max_{\phi} \expt_{x \leftarrow \distr}{\norm{\pi_x \phi}^2}
  \, ,
  \quad 
  \text{where } \pi_x := \sum_{i : x_i = 1} \ket{i}\bra{i} \, .
\end{align*}
Let $\phi := \sum_{i = 1}^m \alpha_i \ket{i}$, with
$\norm{\alpha} = 1$. We have
\begin{align*}
  \solnl = & \max_{\alpha: \norm{\alpha} = 1} \expt_{x \leftarrow
             \distr}{\sum_{i : x_i = 1} \alpha_i^2} \\
  = & \max_{\alpha: \norm{\alpha} =
      1} \sum_{i = 1}^m \alpha_i^2 \cdot \sum_{x \in \bool^m} d_x
      \cdot  x_i \\
  = &  \max_{\alpha: \norm{\alpha} =
      1} \sum_{i = 1}^m \alpha_i^2 \cdot \weight_i  \, ,
\end{align*}
where for each $i\in [m]$, we define 
\begin{align*}
  \weight_i := \sum_{x \in \bool^m} d_x \cdot x_i \, .
\end{align*}
In other words, $\weight_i$ captures the likelihood that $x_i$ is assigned to $1$
under $\distr$.
Then it becomes clear that the maximum is achieved by a vector
$\alpha$ having 0 entries except taking $1$ on $i^*$ where
$\weight_{i^*}$ is maximized. Therefore, 
\begin{align}
  \label{eq:nu}
  \solnl = \weight_{i^*} = \max_{i \in [m]} \weight_i \, .
\end{align}

This matches our earlier analysis in the special
cases. We also note that for any $i\in [m]$, $\weight_i / \weight$,
where $\weight := \sum_i \weight_i$, can be viewed as the probability
that $x_i = 1$ when $x$ is sampled according to $\distr$.

\subsection{Quantum Search Algorithm on \texorpdfstring{$\distr$}{}} %
\label{sec:mgrover}

We are now ready to describe our search algorithm. A main distinction
from standard Grover is that the amplitudes in our initial state are
proportional to the weight $\weight_i$ rather than a uniform
superposition. 

\begin{mdframed}[linecolor=black!7, backgroundcolor=black!7]
  \begin{center}
    \textbf{Quantum Search Algorithm $\cA$ for an arbitrary
      distribution $\distr$ }
  \end{center}

  \textbf{Given}: $x \in \bool^m$ drawn from $\distr$ as a black-box
  function.

  \textbf{Goal}: find $i \in [m]$ such that $x_{i} = 1$ making $\qquery$
  quantum queries to $x$.

  \emph{Initialization}: $\cA$ constructs a unitary $U_\distr$ such
  that
\begin{equation*}
  \ket{\phi_0}: = U_\distr \ket{0} = \frac{1}{\sqrt{\weight}}\sum_{i}
  \sqrt{\weight_i} \ket{i} \, .  
\end{equation*}

\emph{Modified Grover iteration}: Repeatedly apply $G:= R_0R_x$, where
\begin{align*}
  R_0 &:= - (\identity - 2\ket{\phi_0}\bra{\phi_0}) \, ,\\
  R_x \ket{i} &:= (-1)^{x_i} \ket{i} \, .
\end{align*}

\emph{Output}: Measure the state in the computational basis and output
the measurement outcome $i$.

\end{mdframed}

Note that once $U_\distr$ is available
$R_0 = - U_\distr (\identity - \ket{0}\bra{0})U_\distr^\dagger$ can be
implemented easily, and one application of $R_x$ can be realized by
\emph{one} query to $x$. 

For any fixed $x$, we let $\varepsilon_x$ denote the probability that
$\cA$ finds a solution (i.e,. some $i$ with $x_i=1$), and
$\varepsilon:= \expt_{x\gets \distr} \varepsilon_x$ represents the
success probability of $\cA$ averaged over the distribution $\distr$.

\begin{theorem} \label{thm:general_quantum_success}
Algorithm $\cA$ with $\qquery$ quantum queries finds an $i$
  with $x_i = 1$ with probability:
  \begin{align*}
    \varepsilon \ge \qquery^2 \cdot \frac{\sum_i \weight_i^2}{\weight} \, .
  \end{align*}
\label{thm:dgrover}
\end{theorem}
\begin{proof}
  We adapt the geometric analysis of standard Grover's algorithm to
  analyze $\cA$. First for any $x$, define two states below:
  \begin{align*}
    \ket{A_x} : = \frac{1}{\sqrt{\alpha_x}} \sum_{i:x_i =1}
    \sqrt{\weight_i}\ket{i} \, , \quad  
    \ket{B_x} : = \frac{1}{\sqrt{\beta_x}} \sum_{i: x_i = 0} \sqrt{\weight_i} \ket{i} \, ,
  \end{align*}
  with normalization factors
  \begin{align*}
    \alpha_x := \sum_{i: x_i =1} \weight_i = \sum_i
    \weight_i x_i, \quad \text{and} \quad 
    \beta_x : = \sum_{i: x_i =0} \weight_i = \sum_i \weight_i(1-x_i) \, .
  \end{align*}
  We will focus on the two dimensional plane spanned by $\ket{A_x}$
  and $\ket{B_x}$. Observe that $\phi_0$ belongs to this plane, and
  can be decomposed under the basis $\{\ket{A_x},\ket{B_x}\}$:
  \begin{align*}
    \ket{\phi_0} := \sin\theta \ket{A_x} + \cos\theta \ket{B_x} \, ,
  \end{align*}
  where
  \begin{align*}
    \sin^2\theta = \left| \langle \phi_0 | A_x\rangle \right|^2 =
    \frac{1}{\weight \cdot \alpha_x} (\sum_i \weight_i x_i)^2 = \frac{\alpha_x}{\weight}   \, . 
  \end{align*}
  We then show that on the two dimensional plane, $R_0$ is a
  reflection about $\ket{\phi_0}$ and $R_x$ is a reflection
  $\ket{B_x}$. We introduce a state $\ket{\phi_0^\perp}$ on the plane
  orthogonal to $\ket{\phi_0}$, which can be written as
  \begin{align*}
    \ket{\phi_0^\perp} = \cos\theta \ket{A_x} - \sin\theta\ket{B_x} \, .
  \end{align*}
  Clearly $\{\phi_0, \phi_0^\perp\}$ forms another basis on the plane,
  under which we can express $\ket{A_x}$ and $\ket{B_x}$ as below.
  \begin{align*}
    \ket{A_x} & = \sin\theta\ket{\phi_0} + \cos\theta\ket{\phi_0^\perp}\, , \\
    \ket{B_x} & = \cos\theta\ket{\phi_0} - \sin\theta\ket{\phi_0^\perp}   \, .
  \end{align*}
  It then becomes easy to verify that
  \begin{align*}
    R_0 \ket{A_x} & = \sin\theta\ket{\phi_0} -
                    \cos\theta\ket{\phi_0^\perp} \, , \\
    R_0 \ket{B_x} & = \cos\theta\ket{\phi_0} +
                    \sin\theta\ket{\phi_0^\perp} \, .
  \end{align*}
  Hence $R_0$ reflects about $\phi_0$. Similarly, $R_x$ reflects about
  $\ket{B_x}$ as can be seen below.
  \begin{align*}
    R_x \ket{\phi_0} & = - \sin\theta\ket{A_x} +
                    \cos\theta\ket{B_x} \, , \\
    R_0 \ket{\phi_0^\perp} & = -\sin\theta\ket{\phi_0} -
                    \cos\theta\ket{\phi_0^\perp} \, .
  \end{align*}
  As a consequence, $G = R_0R_x$ composes two reflections and
  effectively amounts to an rotation of $2\theta$.
  
  Therefore, after $\qquery$ iterations, the state becomes
  \begin{align*}
    \ket{\phi_{\qquery}} := \sin((2\qquery + 1)\theta) \ket{A_x} +
    \cos((2\qquery + 1)\theta)\ket{B_x} \, .
  \end{align*}
  This is illustrated in Figure~\ref{fig:grover}.
  \begin{figure}[h!]
    \centering
    \includegraphics[width= .4\textwidth]{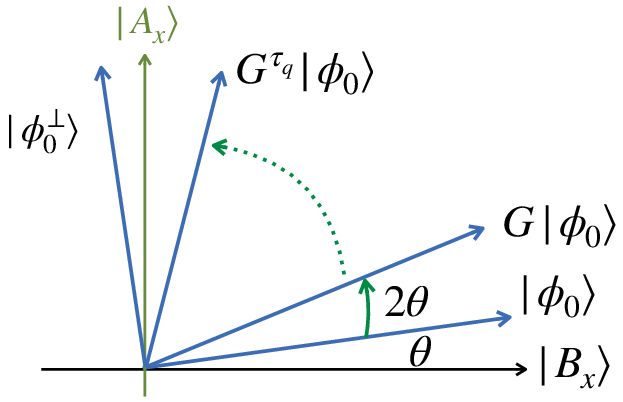}
    \caption{Illustration of the evolution in the two-dimensional
      plane.}
    \label{fig:grover}
  \end{figure}
  When measuring $\ket{\phi_{\qquery}}$, an outcome $i$ with $x_i = 1$
  occurs with probability
  \begin{align*}
    \varepsilon_x = \sin^2((2\qquery + 1)\theta) \ge \left(\frac{2 \qquery + 1}{2}\theta \right)^2
    \ge \qquery^2 \sin^2 \theta = \qquery^2 \frac{\alpha_x}{\weight} \, . 
  \end{align*}
  Thus
  \begin{align*}
    \varepsilon  = \expt_{x\gets \distr} \varepsilon_x  \ge \qquery^2
    \frac{\expt_x\alpha_x}{\weight} = \qquery^2 \frac{\sum_i \weight_i \sum_x d_x x_i}{\weight} 
    = \qquery^2 \frac{\sum_i \weight_i^2}{\weight} \, . 
  \end{align*}
\end{proof}

\paragraph{Optimality for permutation-invariant distributions.}
Consider a special family of distributions, where $\weight_i$ are
identical for all $i \in [m]$ implying that every $i$ is mapped to $1$
with equal probability. We call such a distribution $\distr$
\emph{permutation invariant}, and in this case our quantum algorithm
$\cA$ becomes identical to the standard Grover's algorithm. It also
follows immediately~\cref{eq:nu} that for any $i, \weight_i =
\solnl$. Therefore we obtain that
\begin{align*}
  \frac{\sum_i \weight_i^2}{\weight} =   \frac{\sum_i
    \weight_i^2}{\sum_i \weight_i} = \frac{m \solnl^2}{m \solnl} =
  \solnl \, . 
\end{align*}
As a result, the quantum algorithm $\cA$ succeeds with probability
$\Omega(\qquery^2\solnl)$ in the case of permutation-invariant distribution,
which is in turn \emph{optimal} by our hardness bound
(\cref{thm:main}). This also reproves the tight quantum query
complexity for multi-uniform search and Bernoulli search. We summarize
it below.

\begin{corollary} For a permutation-invariant distribution $\distr$,
  the quantum algorithm $\cA$ coincides with the standard Grover's
  algorithm, and it succeeds with probability
  $\Omega(\qquery^2 \cdot \solnl)$ with $\qquery$ quantum queries which is
  \emph{tight}.

  In particular, multi-uniform search and Bernoulli search have tight
  quantum query complexity $\Theta(\qquery^2\frac{w}{m})$ and
  $\Theta(\qquery^2\bpara)$ for quantum algorithms with $\qquery$ queries.
  \label{cor:optimalgrover}
\end{corollary}

\section{A Hybrid Algorithm for Distributional Search}
\label{sec:hybrid_algorithm}
We are now ready to describe a hybrid algorithm equipped with $\cquery$ classical queries
and $\qquery$ quantum queries.
The idea is simple: Given the distribution $\distr$, let
$S = \{i_1, ..., i_{\cquery}\} \subseteq [m]$ be the set of indices
with the $\cquery$ largest values of $\weight_i$. In case of ties, we break them
arbitrarily. Our algorithm will first issue the $\cquery$ classical
queries on $S$ to verify if there exists any $i\in S$ such that $x_i = 1$. If not, we
run the quantum search algorithm $\cA$ from before, but on the reduced search space $[m] - S$.

To run the quantum algorithm in a modular fashion, we can define an
induced distribution $\tilde D$ on $\bool^{m - \cquery}$. 
We will denote here by $x_{T}$ the substring of $x$ of size $|T|$ obtained from concatenating the bits $x_i$ for all $i \in T$  
and we denote by $\bar S$ the set defined as $\bar S=[m]\backslash S$. 

To define $\newd{\distr}$, we first define
$d:= \sum_{x\in\bool^m: x_S =0} d_x$. Then for each
$\textbf{x} \in \bool^{m - \cquery}$,
we 
will define
$\newd{d}_{\textbf{x}}: = \frac{d_x}{d}$,
where $x$ is the unique string
with $x_S = 0$ and 
$x_{\bar S}= \textbf{x}$.
Note that there is a fixed
mapping that matches every index 
$\textbf{i} \in \bar S$ with an index $i \in [m]$
such that
$\textbf{x}_{\textbf{i}} = 1$
if and only if $x_i = 1$.  We assume this
mapping is performed implicitly whenever necessary. Therefore for
every 
$\textbf{i}\in \bar S$,
we can write the weight under
$\newd{\distr}$,
\begin{align*}
  \newd{\weight}_{\textbf{i}}= \sum_{\textbf{x} \in \bool^{m-\cquery}}
  \newd{d}_{\textbf{x}} \cdot {\textbf{x}}_{\textbf{i}} = \frac{\sum_{x : x_S =0} d_x
  \cdot x_i}{\sum_{x:x_S = 0} d_x} \, . 
\end{align*}

Our hybrid algorithm can now be described 
as follows.
\begin{mdframed}[linecolor=black!7, backgroundcolor=black!7]
  \begin{center}
    \textbf{Hybrid Search Algorithm $\cA_h$ for an arbitrary
      distribution $\distr$ }
  \end{center}

  \textbf{Given}: $x \in \bool^m$ drawn from $\distr$ as a black-box
  function.

  \textbf{Goal}: find $i \in [m]$ such that $x_{i} = 1$ making
  $\cquery$ classical queries $\qquery$ and quantum queries to $x$.

  \emph{Classical Stage.} $\cA$ makes classical queries for each
  $i \in S$, where $S$ as defined above consists of the indices with
  the $\cquery$ largest $\weight_i$. If some $x_i=1$, output $i$
  and abort. Otherwise, continue. 
  
  \emph{Quantum Stage.} Run the quantum algorithm $\cA$ on the induced
  distribution $\newd{\distr}$. 

\end{mdframed}

The success probability can be split into analyzing the classical and
quantum stages separately as we show below.  First, we define the following
binary random variables:
\begin{itemize}
\item $Z_c^x = 1$ if and only if $x_i=1$ for some $i\in S$ (the
  classical stage succeeds);
\item $Z_q^x = 1$ if and only if the quantum stage is successful.
\end{itemize}

\begin{lemma} \label{lemma:hybrid_lb} For any distribution $\distr$,
  the probability that hybrid algorithm $\cA_h$ succeeds is
  \begin{align*}
    \Pr[\mathsf{Hybrid \ Success}] \geq \frac{1}{2} \left(\expt_{x
    \leftarrow \distr}[Z_c^x] + \expt_{x \leftarrow \distr}[Z_q^x]
    \right) \, .
  \end{align*}
\end{lemma}

\begin{proof} The algorithm fails if both classical and quantum stages
  fail. Hence the failure probability is
  \begin{align*}
    \expt_{x \gets \distr} [(1 - Z_c^x)(1 - Z_q^x)] = 1 - \expt_{x \gets \distr}[Z_c^x] - \expt_{x \gets
    \distr}[Z_q^x] + \expt_{x \gets \distr}[Z_c^x \cdot Z_q^x] \, .
  \end{align*}
  Then by using the Cauchy-Schwartz inequality (\cref{lemma:cs}) and as
  $Z_c^x$ and $Z_q^x$ are both binary variables, we have
  \begin{align*}
    \expt_{x \gets \distr}[Z_c^x \cdot Z_q^x] & \le \sqrt{\expt_{x
    \gets \distr}[Z_c^x] \cdot \expt_{x \gets \distr}[Z_q^x]} \\
    & \le \frac{1}{2}(\expt_{x
    \gets \distr}[Z_c^x] + \expt_{x \gets \distr}[Z_q^x]) \, .
  \end{align*}
  We can then conclude that the algorithm's success probability is
  \begin{align*}
    \Pr[\mathsf{Hybrid \ Success}] = 1- \expt_{x \gets \distr} [(1 -
    Z_c^x)(1 - Z_q^x)]  \ge \frac{1}{2} \left(\expt_{x
    \leftarrow \distr}[Z_c^x] + \expt_{x \leftarrow \distr}[Z_q^x]
    \right) \, .
  \end{align*}
\end{proof}
By~\cref{thm:dgrover}, we can immediately give an expression for the
quantum success probability. Namely, 
\begin{align*}
  \expt_{x \leftarrow \distr}[Z_q^x] \ge \qquery^2
  \frac{\sum_{\textbf{i} \in \bar S}
  \newd{\weight}_{\textbf{i}}^2}{\sum_{\textbf{i}\in \bar S}
  \newd{\weight}_{\textbf{i}}} \, .
\end{align*}

\subsection{Success Probability for Special Distributions}

We now show that for some special cases the hybrid algorithm above is
optimal. We note that in these cases, the quantum stage actually
coincides with the standard Grover search, and the quantum success
probability can be obtained by the known result. Our analysis can be
viewed as an alternative approach following the general result
of~\cref{thm:dgrover}.

When $x\gets \distr$ assigns a single $i$ with $x_i=1$ uniformly at
random, $\newd{\distr}$ can be seen as the same distribution but
restricting to $x$ with $x_S=0$. For all 
$\textbf{i} \in \bar S$, we have
$\newd{\weight}_{\textbf{i}} = \frac{1}{m-c} $, and hence:
\begin{align*}
  \expt_{x \leftarrow \distr}[Z_q^x] = \qquery^2\cdot \frac{\sum_{\textbf{i}\in \bar S}
  \newd{\weight}_{\textbf{i}}^2}{\sum_{\textbf{i}\in \bar S}
  \newd{\weight}_{\textbf{i}}} = \qquery^2 \frac{1}{m-c}\, .  
\end{align*}

It is also easy to observe that $\expt_{x \leftarrow \distr}[Z_c^x] =
\cquery \frac{1}{m}$.

\begin{lemma}[Uniform Hybrid Query Complexity]
  When $\distr$ is the uniform distribution, our hybrid algorithm
  equipped with $\cquery$ classical queries and $\qquery$ quantum
  queries succeeds with probability 
\begin{align*}
  \Pr[\mathsf{Hybrid \ Success}] \geq \frac{1}{2} \left(\frac{\cquery}{m} + \frac{\qquery^2}{m - \cquery}  \right).
\end{align*}
Modulo constant factors and lower order terms, this
matches the hardness bound, and hence the hybrid query complexity is
$\Theta \left(\frac{1}{m}(\cquery + \qquery^2) \right)$. 
\end{lemma}

Similarly, we can 
obtain a tight bound for the Bernoulli distribution, by
the observation that $\newd{\distr}$ in this case is just another
Bernoulli with the same $\berpar$. Hence 
\begin{align*}
  \expt_{x \leftarrow \distr}[Z_q^x] = \berpar \cdot \qquery^2 \, .
\end{align*}
On the other hand,
\begin{align*}
  \expt_{x \leftarrow \distr}[Z_c^x] = 1 - (1-\berpar)^{\cquery} \ge
  \frac{1}{2}\berpar\cdot \cquery \, . 
\end{align*}

\begin{lemma}[Bernoulli Hybrid Success and Optimality]
When $\distr$ is the Bernoulli distribution, our hybrid algorithm equipped with $\cquery$ classical queries and $\qquery$ quantum queries succeeds with probability at least:
\begin{align*}
  \Pr[\mathsf{Hybrid \ Success}] \ge \frac{1}{2}\eta \left(
  \frac{1}{2} \cquery + \qquery^2 \right)   \, .
\end{align*}
Hence the hybrid query complexity for the Bernoulli distribution is
$\Theta(\berpar(\cquery + \qquery^2))$.
\end{lemma}

\section*{Acknowledgements} J.G. was partially supported by NSF grants no. 2001082 and 2055694. F.S. was partially supported by NSF grant no. 1942706 (CAREER). J.G. and F.S. 
were also partially support by Sony by means of the Sony Research Award Program. 
A.C. acknowledges support from the National Science Foundation grant CCF-1813814 and from the AFOSR under Award Number FA9550-20-1-0108.

{\small 
\bibliographystyle{alpha}
\bibliography{biblio}

\newcommand{\etalchar}[1]{$^{#1}$}
\begin{thebibliography}{ABKM22}

\bibitem[ABKM22]{ABKM22}
Gorjan Alagic, Chen Bai, Jonathan Katz, and Christian Majenz.
\newblock Post-quantum security of the even-mansour cipher.
\newblock In {\em Advances in Cryptology -- EUROCRYPT 2022}, pages 458--487.
  Springer, 2022.

\bibitem[AHU19]{AHU19}
Andris Ambainis, Mike Hamburg, and Dominique Unruh.
\newblock Quantum security proofs using semi-classical oracles.
\newblock In {\em Advances in Cryptology -- CRYPTO 2019}, pages 269--295.
  Springer, 2019.

\bibitem[AMRS20]{AMRS20}
Gorjan Alagic, Christian Majenz, Alexander Russell, and Fang Song.
\newblock Quantum-secure message authentication via blind-unforgeability.
\newblock In {\em Advances in Cryptology -- EUROCRYPT 2020}. Springer, 2020.

\bibitem[ARU14]{ARU14}
Andris Ambainis, Ansis Rosmanis, and Dominique Unruh.
\newblock Quantum attacks on classical proof systems: The hardness of quantum
  rewinding.
\newblock In {\em 2014 IEEE 55th Annual Symposium on Foundations of Computer
  Science}, pages 474--483. IEEE, 2014.

\bibitem[BBBV97]{BBBV97}
Charles~H Bennett, Ethan Bernstein, Gilles Brassard, and Umesh Vazirani.
\newblock Strengths and weaknesses of quantum computing.
\newblock {\em SIAM journal on Computing}, 26(5):1510--1523, 1997.

\bibitem[BDF{\etalchar{+}}11]{BDF+11}
Dan Boneh, {\"O}zg{\"u}r Dagdelen, Marc Fischlin, Anja Lehmann, Christian
  Schaffner, and Mark Zhandry.
\newblock Random oracles in a quantum world.
\newblock In {\em Advances in Cryptology -- ASIACRYPT 2011}, pages 41--69.
  Springer, 2011.

\bibitem[BR93]{BR93}
Mihir Bellare and Phillip Rogaway.
\newblock Random oracles are practical: A paradigm for designing efficient
  protocols.
\newblock In {\em Proceedings of the 1st ACM conference on Computer and
  Communications Security}, pages 62--73, 1993.

\bibitem[BR94]{BR94}
Mihir Bellare and Phillip Rogaway.
\newblock Optimal asymmetric encryption.
\newblock In {\em Advances in Cryptology--EUROCRYPT 1994}, pages 92--111.
  Springer, 1994.

\bibitem[BR96]{BR96}
Mihir Bellare and Phillip Rogaway.
\newblock The exact security of digital signatures-how to sign with rsa and
  rabin.
\newblock In {\em Advances in Cryptology--Eurocrypt 1996}, pages 399--416.
  Springer, 1996.

\bibitem[BZ13]{BZ13}
Dan Boneh and Mark Zhandry.
\newblock Secure signatures and chosen ciphertext security in a quantum
  computing world.
\newblock In {\em Advances in Cryptology -- CRYPTO 2013}, pages 361--379.
  Springer, 2013.

\bibitem[CCHL22]{CCHL22}
Sitan Chen, Jordan Cotler, Hsin-Yuan Huang, and Jerry Li.
\newblock The complexity of nisq, 2022.

\bibitem[CCL23]{CCL23}
Nai-Hui Chia, Kai-Min Chung, and Ching-Yi Lai.
\newblock On the need for large quantum depth.
\newblock {\em J. ACM}, 70(1), jan 2023.

\bibitem[CEV23]{CEV23}
C{\'e}line Chevalier, Ehsan Ebrahimi, and Quoc-Huy Vu.
\newblock On security notions for encryption in a quantum world.
\newblock In {\em Progress in Cryptology -- INDOCRYPT 2022}, pages 592--613.
  Springer, 2023.

\bibitem[CGK{\etalchar{+}}23]{CGKSW23}
Alexandru Cojocaru, Juan Garay, Aggelos Kiayias, Fang Song, and Petros Wallden.
\newblock Quantum {M}ulti-{S}olution {B}ernoulli {S}earch with {A}pplications
  to {B}itcoin's {P}ost-{Q}uantum {S}ecurity.
\newblock {\em {Quantum}}, 7:944, 2023.

\bibitem[CM20]{CM20}
Matthew Coudron and Sanketh Menda.
\newblock Computations with greater quantum depth are strictly more powerful
  (relative to an oracle).
\newblock In {\em Proceedings of the 52nd Annual ACM SIGACT Symposium on Theory
  of Computing}, STOC 2020, page 889–901, New York, NY, USA, 2020.
  Association for Computing Machinery.

\bibitem[CMS19]{CMS19}
Alessandro Chiesa, Peter Manohar, and Nicholas Spooner.
\newblock Succinct arguments in the quantum random oracle model.
\newblock In {\em 17th International Theory of Cryptography Conference -- TCC
  2019}, pages 1--29. Springer, 2019.

\bibitem[DFMS19]{DFMS19}
Jelle Don, Serge Fehr, Christian Majenz, and Christian Schaffner.
\newblock Security of the {Fiat-Shamir} transformation in the quantum
  random-oracle model.
\newblock In {\em Advances in Cryptology -- CRYPTO 2019}, pages 356--383.
  Springer, 2019.

\bibitem[DFMS22]{DFMS22}
Jelle Don, Serge Fehr, Christian Majenz, and Christian Schaffner.
\newblock Online-extractability in the quantum random-oracle model.
\newblock In {\em Advances in Cryptology -- EUROCRYPT 2022}, pages 677--706.
  Springer, 2022.

\bibitem[DH09]{DH09}
C{\u{a}}t{\u{a}}lin Dohotaru and Peter H{\o}yer.
\newblock Exact quantum lower bound for grover's problem.
\newblock {\em Quantum Information \& Computation}, 9(5):533--540, 2009.

\bibitem[ES15]{ES15}
Edward Eaton and Fang Song.
\newblock {Making Existential-unforgeable Signatures Strongly Unforgeable in
  the Quantum Random-oracle Model}.
\newblock In {\em 10th Conference on the Theory of Quantum Computation,
  Communication and Cryptography -- TQC 2015}, volume~44 of {\em Leibniz
  International Proceedings in Informatics (LIPIcs)}, pages 147--162. Schloss
  Dagstuhl--Leibniz-Zentrum fuer Informatik, 2015.

\bibitem[ES20]{ES20}
Edward Eaton and Fang Song.
\newblock A note on the instantiability of the quantum random oracle.
\newblock In {\em International Conference on Post-Quantum Cryptography}, pages
  503--523. Springer, 2020.

\bibitem[FO13]{FO-JoC13}
Eiichiro Fujisaki and Tatsuaki Okamoto.
\newblock Secure integration of asymmetric and symmetric encryption schemes.
\newblock {\em Journal of Cryptology}, 26(1):80--101, 2013.
\newblock Preliminary version in CRYPTO 1999.

\bibitem[FOPS04]{FOPS04}
Eiichiro Fujisaki, Tatsuaki Okamoto, David Pointcheval, and Jacques Stern.
\newblock {RSA-OAEP} is secure under the rsa assumption.
\newblock {\em Journal of Cryptology}, 17(2):81--104, 2004.
\newblock Preliminary version in CRYPTO 2001.

\bibitem[GKL15]{GKL15}
Juan Garay, Aggelos Kiayias, and Nikos Leonardos.
\newblock The bitcoin backbone protocol: Analysis and applications.
\newblock In {\em Advances in Cryptology -- EUROCRYPT 2015}, pages 281--310.
  Springer, 2015.

\bibitem[Gro96]{Grover96}
Lov~K Grover.
\newblock A fast quantum mechanical algorithm for database search.
\newblock In {\em Proceedings of the twenty-eighth annual ACM symposium on
  Theory of computing}, pages 212--219. ACM, 1996.

\bibitem[HHK17]{HHK17}
Dennis Hofheinz, Kathrin H{\"o}velmanns, and Eike Kiltz.
\newblock A modular analysis of the fujisaki-okamoto transformation.
\newblock In {\em 15th International Theory of Cryptography Conference -- TCC
  2017}, pages 341--371. Springer, 2017.

\bibitem[HLS22]{HLS22}
Yassine Hamoudi, Qipeng Liu, and Makrand Sinha.
\newblock Quantum-classical tradeoffs in the random oracle model, 2022.

\bibitem[HRS16]{HRS16}
Andreas H\"{u}lsing, Joost Rijneveld, and Fang Song.
\newblock Mitigating multi-target attacks in hash-based signatures.
\newblock In {\em 19th IACR International Conference on Public-Key Cryptography
  --- PKC 2016}, pages 387--416. Springer, 2016.

\bibitem[JST21]{JST21}
Joseph Jaeger, Fang Song, and Stefano Tessaro.
\newblock Quantum key-length extension.
\newblock In {\em 19th International Theory of Cryptography Conference -- TCC
  2021}, pages 209--239. Springer, 2021.

\bibitem[KM10]{KM10}
Hidenori Kuwakado and Masakatu Morii.
\newblock Quantum distinguisher between the 3-round feistel cipher and the
  random permutation.
\newblock In {\em 2010 IEEE International Symposium on Information Theory},
  pages 2682--2685. IEEE, 2010.

\bibitem[Pre18]{Preskill18}
John Preskill.
\newblock Quantum computing in the {NISQ} era and beyond.
\newblock {\em Quantum}, 2:79, 2018.

\bibitem[Ros22]{Ros22}
Ansis Rosmanis.
\newblock Hybrid quantum-classical search algorithms.
\newblock {\em arXiv preprint arXiv:2202.11443}, 2022.

\bibitem[Sho01]{Shoup01}
Victor Shoup.
\newblock {OAEP} reconsidered.
\newblock In {\em Advances in Cryptology—-CRYPTO 2001}, pages 239--259.
  Springer, 2001.

\bibitem[SZ19]{SZ19}
Xiaoming Sun and Yufan Zheng.
\newblock Hybrid decision trees: Longer quantum time is strictly more powerful,
  2019.

\bibitem[Unr15]{Unruh15}
Dominique Unruh.
\newblock Non-interactive zero-knowledge proofs in the quantum random oracle
  model.
\newblock In {\em Advances in Cryptology -- EUROCRYPT 2015}, pages 755--784.
  Springer, 2015.

\bibitem[YZ21]{YZ21}
Takashi Yamakawa and Mark Zhandry.
\newblock Classical vs quantum random oracles.
\newblock In {\em Advances in Cryptology -- EUROCRYPT 2021}, pages 568--597.
  Springer, 2021.

\bibitem[Zal99]{Zalka99}
Christof Zalka.
\newblock Grover's quantum searching algorithm is optimal.
\newblock {\em Physical Review A}, 60(4):2746, 1999.

\bibitem[Zha15]{Zhandry15_ibe}
Mark Zhandry.
\newblock Secure identity-based encryption in the quantum random oracle model.
\newblock {\em International Journal of Quantum Information}, 13(04):1550014,
  2015.
\newblock Preliminary version in IACR CRYPTO 2012.

\bibitem[Zha19]{Zhandry19}
Mark Zhandry.
\newblock How to record quantum queries, and applications to quantum
  indifferentiability.
\newblock In {\em Advances in Cryptology -- CRYPTO 2019}, pages 239--268.
  Springer, 2019.

\bibitem[Zha21]{Zhandry21_qprf}
Mark Zhandry.
\newblock How to construct quantum random functions.
\newblock {\em Journal of the ACM (JACM)}, 68(5):1--43, 2021.
\newblock Preliminary version in FOCS 2012.

\end{thebibliography}
} 

\end{document}